\documentclass[11pt]{article}
\usepackage{etex}
\usepackage[hidelinks]{hyperref}
\usepackage{fontenc}
\usepackage[latin1]{inputenc}
\usepackage{amsmath}
\usepackage{amsfonts}
\usepackage{amssymb}
\usepackage[affil-it]{authblk}
\usepackage{amsthm}
\usepackage{fancyhdr}
\usepackage{enumerate}
\usepackage{mathtools}
\usepackage{mathdots}
\usepackage{booktabs}
\usepackage{multirow}
\usepackage{tikz}
\usepackage{graphicx}
\usepackage{makeidx}
\usepackage{subcaption}
\usepackage{adjustbox}
\usepackage{geometry}
\usepackage{fullpage}

\usetikzlibrary{arrows.meta, positioning}

\usepackage{color}
\usepackage{caption}
\usepackage{float}
\restylefloat{table}
\usepackage{bbm}
\usepackage{bm}
\usepackage{xcolor}
\usepackage{tikz}
\usetikzlibrary{shapes,decorations,arrows,calc,arrows.meta,fit,positioning}

\usepackage{titlesec}
\usepackage{algorithm}
\usepackage{algpseudocode}

\usepackage{enumitem, color, amssymb}
\usepackage{array}
\usepackage{makecell}
\usepackage{subcaption}
\usepackage[square,numbers]{natbib}
\setcitestyle{authoryear,open={(},close={)}}

\title{Outcome-Informed Weighting for Robust ATE Estimation}
\author[1]{Linying Yang\thanks{Corresponding email: linying.yang@stats.ox.ac.uk}}
\author[1]{Robin J.~Evans}
\affil[1]{Department of Statistics, University of Oxford}
\date{\today}
\usepackage{cleveref}

\providecommand{\keywords}[1]{\textit{Key words:} #1}
\theoremstyle{plain}
\newtheorem{theorem}{Theorem}

\newtheorem{proposition}{Proposition}

\newtheorem{example}{Example}

\newtheorem{assumption}{Assumption}
\newtheorem{definition}{Definition}
\newtheorem{remark}{Remark}

\theoremstyle{remark}

\DeclareMathOperator{\Var}{Var}
\DeclareMathOperator{\Cov}{Cov}

\renewcommand{\P}{\mathbbm{P}}

\begin{document}


\maketitle

\newcommand\floor[1]{\lfloor#1\rfloor}
\newcommand\ceil[1]{\lceil#1\rceil}
\newcommand\ind{\protect\mathpalette{\protect\independenT}{\perp}}
\def\independenT#1#2{\mathrel{\rlap{$#1#2$}\mkern2mu{#1#2}}}

\newcommand{\bigb}[2]{\left(\frac{#1}{#2}\right)}
\newcommand{\E}{\mathbb{E}}

\newcommand{\Ysbar}{\bar{Y}_{S}}
\newcommand{\xra}[1]{\overset{#1}{\;\rightsquigarrow\;}}

\newcommand{\linying}[1]{{\textcolor{teal}{{(Linying:} #1)}}}
\newcommand{\robin}[1]{{\textcolor{blue!50!black}{{[Robin:} #1]}}}

\begin{abstract}
Reliable causal effect estimation from observational data requires  adjustment for confounding and sufficient overlap in covariate distributions between treatment groups. However, in high-dimensional settings, lack of overlap often inflates the variance and weakens the robustness of inverse propensity score weighting (IPW) based estimators. Although many approaches that rely on covariate adjustment have been proposed to mitigate these issues, we instead shift the focus to the outcome space. In this paper, we introduce the Augmented Marginal outcome density Ratio (AMR) estimator, an outcome-informed weighting method that naturally filters out irrelevant information, alleviates practical positivity violations and outperforms standard augmented IPW and covariate adjustment-based methods in terms of both efficiency and robustness. Additionally, by eliminating the need for strong a priori assumptions, our post-hoc calibration framework is also effective in settings with high-dimensional covariates. We present experimental results on synthetic data, the NHANES dataset and text applications, demonstrating the robustness of AMR and its superior performance under weak overlap and high-dimensional covariates.

\keywords causal inference, average treatment effect estimation, efficient influence function, double robustness, weak overlap.
\end{abstract}

\section{Introduction}

Doubly robust estimators, such as Augmented Inverse Propensity Score Weighting (AIPW; \citealp{robins1994estimation}), offer the advantage of providing ``two chances'' to perform estimation correctly and still obtain a consistent estimator. Specifically, this means that the estimator of causal parameter of interest---in the context of this paper, the average treatment effect (ATE)---remains consistent if at least one of the propensity score and outcome regression models is correctly specified. However, due to inverse probability weighting by the propensity score, these estimators can suffer from practical positivity violation \citep{petersen2012diagnosing}, where some covariates predict the treatment so well that our estimated weights become extremely large; this inflates the efficiency bound and estimation variance \citep{robins2007comment}.

Not all information needs to be included in the adjustment via a propensity score. For instance, instrumental variables---that is, a variable that only affects the treatment---do not aid in estimating treatment effects, but these may substantially increase estimation variance. A rule of thumb that we can take from recent work by \citet{henckel22graphical} and \citet{rotnitzky2020efficient} is that we should only include variables that predict the \emph{outcome} well. Essentially, we are concerned with evaluating the differences in outcomes resulting from treatment interventions, and controlling the variation in those observed outcomes.

This leads to the concept of the marginal density ratio \citep{taufiq2024marginal}. Instead of manually evaluating propensity scores or selecting features in the pre-treatment covariate space, \emph{we shift our focus to the outcome space}, allowing the observed outcomes to \textit{guide the weighting} and determine which information should be included. In this paper, we take the Marginal outcome density Ratio estimator (MR) of \citet{taufiq2024marginal}, and introduce the Augmented Marginal outcome density Ratio estimator (AMR), which attains double robustness and asymptotic normality under certain regularity conditions; we demonstrate the advantages of these estimators in filtering necessary information, obtaining treatment effects more effectively both in small samples and asymptotically than their direct counterparts, IPW and AIPW.

Although numerous methods have been proposed to address complex covariate structures, they typically rely on imposing \emph{a priori} assumptions. While such assumptions can improve estimation efficiency, they often incur the costs of model misspecification and error propagation. In this paper, we show that rather than depending on a priori covariate selection or representation learning, MR and AMR achieve reduced estimation variance through \emph{post-hoc} information filtering in their design. To our knowledge, no prior work has contrasted methods based on a priori assumptions with those using post-hoc calibration.

The paper is organized as follows. We begin by formulating the problem and introducing the relevant notations in \Cref{sec:background}. In this section, we also review existing estimation methods---including re-weighting and balancing---which fundamentally involve selecting adjustment sets or learning representations to address the challenges of high-dimensional problems and lack of overlap, motivations that are shared with our proposed approaches. Formal definitions of MR and AMR are introduced in \Cref{sec:method}. Next, we describe the rationale for incorporating outcome information into the weighting process, which underlies our estimators. We illustrate the benefits of leveraging outcome information regarding distilling essential information in \Cref{sec:outcome_driven_weighting}. In \Cref{subsec:connect_efficient_set}, we discuss how the proposed method relates to current work on learning efficient adjustment sets, highlighting the flexibility and robustness of our outcome-informed, post-hoc univariate adjustment. We also provide a detailed explanation on how MR and AMR smooth the weights thus reducing the variance in the same section.

We then turn our attention to the asymptotic properties of our proposed estimators, especially detailing AMR's double robustness and asymptotic behavior in \Cref{sec:asymptotic} as an enhancement over MR. Since achieving asymptotic normality requires that the univariate regression used for weight estimation converges at specific rates---and given that our approach introduces an additional layer of weight estimation---we also provide guidance on when regression-based techniques are particularly effective in attaining the desired properties in \Cref{sec:weights_estimation}.

In \Cref{subsec:synthetic_experiment}, we conduct synthetic experiments that demonstrate the superior performance of AMR as a doubly robust, post-hoc adjustment method in settings with either high-dimensional covariate or lack of overlap. In \Cref{subsec:experiment_real}, we illustrate our methods on the NHANES dataset. Recognizing the strength of MR and AMR in high-dimensional contexts, we further compare our estimators against alternative approaches for causal effect estimation on the News and Amazon review data, where the text is used as a collection of covariates, in \Cref{subsec:experiment_text} and \Cref{sec:amazon}. 

By capitalizing on outcome information, our approach opens up multiple avenues for applications and extensions. We conclude in \Cref{sec:discussion} by outlining several potential extensions and suggesting directions for future research.

\section{Background}\label{sec:background}
\subsection{Problem set-up and notation}

We use $\P$ to denote the true distribution, and $\P_n$ for the empirical measure, $\E$ represents an expectation, and $\E f= \int f \mathrm{d}\P$.
Sample averages are written as $\P_n[f] = \frac{1}{n}\sum_{i=1}^n f(Z_i)$. We use $\rightsquigarrow$ to denote convergence in distribution, and $\overset{P}{\longrightarrow}$ to denote convergence in probability. We write $\left\lVert f\right\rVert_2 = \left[\int f(z)^2 \mathrm{d}\P(z)\right]^{\frac{1}{2}}$ for the $L_2(\P)$ norm of a (possibly) random function $f$. We adapt the usual symbols $\operatorname{O}_\P(\cdot)$ and $\operatorname{o}_\P(\cdot)$; by writing $X_n = \operatorname{O}_\P(a_n)$, we mean $X_n/a_n$ is bounded in $\P$-probability, and $X_n = \operatorname{o}_\P(a_n)$ means $X_n/a_n \overset{P}{\longrightarrow}0$ under $\P$.

Assume we have observations $Z=(X,A,Y)$, where $X$ is a vector of $p$ pre-treatment covariates with support $\mathcal{X}$, $A \in \{0,1\}$ is a binary treatment indicator, and $Y$ is the observed outcome. We assume the joint distribution of $(X,A,Y)$, denoted by $\P$, belongs to a collection of probability measures $\mathcal{P}$. We use the potential outcomes framework (\citealt{neyman1923applications} and \citealt{rubin1974estimating}) and write the outcome under a treatment $a \in \{0,1\}$ as $Y^a$. We denote the conditional mean of these by $\mu^a(X) := \E\left[Y^a\middle| X\right]$, for $a\in \{0,1\}$ and the propensity score by $\pi(X):=P(A=1\mid X)$. We sometimes omit the arguments and simply write $\mu^a$, $\pi$ rather than $\mu^a(X)$, $\pi(X)$ when these are clear from the context. Similarly, for their estimates $\hat{\mu}^a(X)$, $\hat{\pi}(X)$ we occasionally write $\hat{\mu}^a$, $\hat{\pi}$ for simplicity. The target parameter of interest in this paper is the average treatment effect (ATE), denoted $\theta = \E\left[Y^1-Y^0\right]$; we write an associated  estimator as $\hat{\theta}$.

Throughout the paper, we take the following standard assumptions as holding:

\begin{assumption}[Identification]\label{ass:common}
\;
\begin{itemize}
    \item \textbf{Positivity:} $0<\pi(x)<1,\quad \forall x \in \mathcal{X}$;
    \item \textbf{Unconfoundedness:} $ \{Y^0,Y^1\} \ind A \mid X$;
    \item \textbf{Consistency:} $A=a$ implies $Y=Y^a$, $a\in\{0,1\}$.
\end{itemize}
\end{assumption}

 Under \Cref{ass:common}, the observed outcome can be expressed as $Y=AY^1+(1-A)Y^0$. The mean of each potential outcome $\E Y^a$ is identified as $ \E \mu^a(X) =  \E_X\E\left[Y\middle|X, A=a\right]$, thus $\theta$ is identifiable as $\theta =\E_X\E\left[Y\middle|X, A=1\right] - \E_X\E\left[Y \middle|X, A=0\right]$.

Next, define $Y^{*} = Y - \mu^{*}(X)$, where $\mu^{*}(X) = \pi(X) \mu^0(X) + \left\{1-\pi(X)\right\}\mu^1(X)$. Similarly, we define the estimated counterparts as $\hat{Y}^{*} = Y - \hat{\mu}^{*}(X)$, where $\hat{\mu}^{*}(X) = \hat{\pi}(X)\hat{\mu}^0(X) + \left\{1-\hat{\pi}(X)\right\}\hat{\mu}^1(X)$. For simplicity, we sometimes use the notation $h(A,X) = \frac{A-\pi(X)}{\pi(X)\left\{1-\pi(X)\right\}}$ and $\hat{h}(A,X) = \frac{A-\hat{\pi}(X)}{\hat{\pi}(X)\left\{1-\hat{\pi}(X)\right\}}$ to represent the so-called `clever covariates' commonly used in the Targeted Maximum Likelihood Estimation (TMLE; \citealp{van2011targeted}) literature. The IPW estimator and AIPW estimator can thus be written as $\P_n \big[\hat{h}(A,X)Y\big]$ and $\P_n\big[\hat{h}(A,X)\hat{Y}^{*}\big]$ respectively.

\subsection{Efficient estimators and efficiency bound of ATE}
With $\theta$ being the target functional and denote the true nuisance parameters as $\xi = (\mu^1,\mu^0,\pi)$, we have its efficient influence function (EIF):
\begin{align}
    \varphi(Z;\theta,\xi) = \frac{A}{\pi(X)}\Bigl\{Y-\mu^1\left(X\right)\Bigr\}   - \frac{1-A}{1-\pi(X)}\Bigl\{Y-\mu^0\left(X\right)\Bigr\} +  \Bigl\{\mu^1(X)-\mu^0(X)\Bigr\} -\theta\;
\label{eqn:eif}
\end{align}
which can also be written as $\varphi(Z;\theta,\xi) = h(A,X)Y^{*}-\theta$. 

AIPW and TMLE provide efficient estimators, developed by solving the estimating equation $\P_n\varphi(Z;\hat{\xi},\hat{\theta})=0$, where $\hat{\xi}=(\hat{\mu}^0,\hat{\mu}^1,\hat{\pi})$. Under certain regularity conditions, these estimators are asymptotically normal 
\begin{align*}
    \sqrt{n}\left(\hat{\theta} - \theta\right) \;\rightsquigarrow\; \mathcal{N}\left(0, \Var\left\{\varphi(Z;\theta,\xi)\right\}\right),
\end{align*}
where $\Var\left\{\varphi\left(Z;\theta,\xi\right)\right\}$ is the nonparametric efficiency bound, so no other regular estimator can have variance smaller than $\Var\left\{\varphi(Z;\theta,\xi)\right\}$ in the asymptotic sense. We refer to \citet{van2000asymptotic, tsiatis2006semiparametric} and 
\citet{kennedy2024semiparametric} for detailed reviews.

\subsection{Related work}\label{subsec:related_work}
The no unmeasured confounding and positivity assumptions are crucial in the literature on average treatment effect estimation. Research has shown that including more confounders predictive of both treatment and control can reduce bias from unmeasured confounders. However, this also increases the risk of positivity violation and, consequently, variance inflation due to the inverse probability weighting format using propensity scores.  We can think of this as a bias-variance trade-off, since including more variables will generally reduce the bias,\footnote{Though not always, as the \emph{M-bias} example shows \citep{greenland1999causal}.} but always improves the predictive power of the propensity score and therefore increases the variance.  In finite samples, the practical positivity violations \citep{petersen2012diagnosing} are especially pronounced when the covariate space is high-dimensional \citep{d2021overlap}.

To address high-dimensional covariates and the resulting practical positivity issues, the literature has explored various strategies for covariate selection. One approach focuses on identifying which covariates to adjust for when their roles are known. For instance, \citet{rosenbaum2002overt} advocate conditioning on all pre-treatment covariates---a recommendation later reinforced by \citet{rubin2009author}. In contrast, \citet{vander2011new} argue for adjusting only for covariates that influence either the treatment or the outcome, \citet{brookhart2006variable} advise adjusting solely for those that affect the outcome, and \citet{pearl2012class} caution that including certain variables, such as instrumental variables, can amplify confounding bias.

An alternative bypasses explicit variable selection based on their roles by directly mitigating extreme propensity scores. The Covariate Balancing Propensity Score (CBPS; \citealp{imai2014covariate}) and its high-dimensional extension, hd-CBPS \citep{ning2020robust}, adjust logistic regression for propensity score estimation via covariate balancing. A significant body of work on Collaborative Targeted Maximum Likelihood Estimation (CTMLE; \citealp{van2010collaborative, diaz2018doubly, ju2019collaborative, benkeser2020nonparametric}) is built on the idea that jointly learning the propensity score and outcome regression models can optimize the bias-variance trade-off. Built on the outcome regression estimation updating procedure in original TMLE, the CTMLE selects the propensity score model from a pool of candidates. Each propensity model on the variable adjustment set is treated as a candidate model for selection and evaluated to see if outcome regression bias is reduced, which is time-consuming. There are also many decision points in the procedure, making the overall estimation strategy rather complicated. It was argued in \citet{van2010collaborative} that if the outcome regression estimation succeeds in explaining most of the true outcome regression, only little inverse propensity weighting is needed. Based on this idea,  \citet{diaz2018doubly} presents a more complex---and consequently less transparent---method, which learns the adapted propensity score (e-score) by first regressing the outcome bias on the estimated propensity score  $r(X) = \E\left[Y-\hat{\mu}^a(X)\,\middle|\,A=a,\hat{\pi}(X)\right]$, then regressing the propensity score on this quantity to obtain $e(X) = \E\left[\hat{\pi}(X)\,\middle|\,r(X)\right]$, and then finally projecting all information onto the covariate space $X$. In another variant, \citet{benkeser2020nonparametric} suggest conditioning on the estimated outcome mean $\hat{\mu}^a(X)$ when its estimation is consistent.

With the advent of neural networks and deep representation learning, recent methods have integrated collaborative learning into representation learning for ATE estimation. \citet{shi2019adapting} proposed DragonNet, which emphasizes incorporating treatment-predictive information even at the expense of outcome prediction performance. This contrasts with approaches aiming to mitigate extreme lack of overlap by sacrificing some predictive accuracy for the propensity score.

A recent contribution by \citet{christgau2024efficient} extends \citet{rotnitzky2020efficient} by representing adjustment variables through a learned representation $V=\eta(X)$. For the estimand $\theta(V) = \E_\P\left[\E\left[Y\middle|V, A=1\right] - \E\left[Y\middle|V, A=0\right]\right]$, the corresponding efficient influence function is
\begin{equation}
    \label{eqn:adjusted_eff_inf}
        \varphi(V) = \frac{A}{\pi(V)}\Bigl\{Y-\mu^1\left(V\right)\Bigr\}   - \frac{1-A}{1-\pi(V)}\Bigl\{Y-\mu^0\left(V\right)\Bigr\} +  \Bigl\{\mu^1(V)-\mu^0(V)\Bigr\} -\theta(V),
\end{equation}
where the nuisance parameters are $\mu^a(V) = \E\left[Y^a\middle |V\right]$, $\pi(V) = \E\left[A \middle|V\right]$. The variance $\Var\left[\varphi(V)\right]$ defines the semiparametric efficiency bound for regular asymptotically linear estimators of $\theta(V)$. Under the condition that $\E_V\E\left[Y\middle|V, A=a\right] = \E_X\E\left[Y\middle|X, A=a\right]$ for all $a\in\{0,1\}$ (defined as $\mathcal{P}$-valid in their paper), it follows that $\theta(X) = \theta(V)$. In such cases, the choice between conditioning on $X$ or on the representation $V$ depends on which yields a smaller efficiency bound.

Suppose the outcome regression factors through an intermediate representation $\eta(\beta,x)$ such that $\mu^a(x) = h(a,\eta(\beta,x))$. \citet{christgau2024efficient} propose the Debiased Outcome-adapted Propensity Estimator (DOPE) using $V_\beta=\eta(\beta,X)$ and demonstrate that under suitable conditions (see Corollary 3.6 in their paper), $\theta(V_{\beta})=\theta(X)$ and more importantly,
 \begin{equation}
 \label{eqn:eff_bd_comparison}
     \Var(\varphi(X)) \geq \Var(\varphi(V_\beta)).
 \end{equation}

We use $\hat{\theta}(X)$ and $\hat{\theta}_\beta$ to denote estimators conditioning on $X$ and $V_\beta$ when $\beta$ is known. The inequality in (\ref{eqn:eff_bd_comparison}) implies that $\hat{\theta}_\beta$ is at least as efficient as $\hat{\theta}(X)$. In practice, however, $\beta$ is unknown and must be estimated, introducing additional uncertainty.

 With a further assumption that the representation takes a single-index form, where the outcome regression factors through the linear predictor as
\begin{equation}\label{eqn:single-index}
    Y=h(\alpha A + \beta^\top X) +\epsilon_Y,
\end{equation}
they build a neural network architecture that learns outcome regression with first layer designed to first learn the outcome regression with a linear first layer enforcing the single-index model. The resulting weights provide an estimate $\hat{\beta}$, which is then used to regress both the outcome and treatment on $\hat{\beta}^\top X$. Although reminiscent of DragonNet's architecture, DOPE focuses on adapting the propensity score and outcome regression to the intermediate representation. 

Notably, \citet{benkeser2020nonparametric} consider a special case of DOPE where the propensity score is estimated conditional on the outcome mean $\hat{\mu}^a(X)$ rather than on the intermediate representation $V_{\hat{\beta}}=\hat{\beta}^\top X$. Similarly, \citet{gui2022causal} propose a transformer-based architecture (the TI estimator) for ATE estimation from text, which shares conceptual similarities with these methods. Moreover, \citeauthor{diaz2018doubly}'s \emph{e-score} builds on this idea by choosing the covariate representation $V=\E\left[Y-\hat{\mu}^a(X)\,\middle|\,\hat{\pi}(X)\right]$ to condition on. 

Many of these methods depend on strong a priori knowledge or assumptions, which are not always available. For example, applying the backdoor criterion for covariate selection \citep{pearl1995causal} requires strong understanding of the true causal dependence structure, while covariate balancing methods typically assume linear relationships \citep{imai2014covariate, athey2018approximate, ning2020robust}. Likewise, representation learning approaches \citep{benkeser2020nonparametric, gui2022causal, christgau2024efficient} depend heavily on a pre-processing step that must accurately capture the underlying representations. In high-dimensional settings, ensuring that these representations are learned effectively is particularly challenging, and any errors introduced during this step can propagate in ways that are difficult to control. Consequently, these methods may have much worse performance when their underlying assumptions are violated. Although \citet{christgau2024efficient} develop semiparametric efficiency theories, their approach is limited to specific representation structures and does not account for the uncertainty introduced by the estimation of the representation learning step, $\hat{\beta}$---not to mention the potential for model misspecification. As demonstrated in \Cref{subsec:synthetic_experiment}, such assumptions can easily lead to substantially biased estimates.

Our proposed estimators aim to stabilize inverse propensity score weighting. While similar in spirit to (smoothed) propensity score trimming \citep{rosenbaum2002overt, crump2009dealing, yang2018asymptotic} and overlap weights \citep{li2018balancing}, our method addresses a distinct issue: we do not alter the target estimand. As detailed in \Cref{sec:asymptotic}, our estimators converge asymptotically to the ATE, unlike trimming methods or the average treatment effect on the overlap (ATO) estimand.

Finally, although high-dimensional covariates are often associated with positivity violations, the two do not always coincide. In some cases, high-dimensionality may not hinder propensity score estimation significantly, whereas lack of overlap can occur even with a few covariates. By leveraging outcome information, our method remains applicable in both scenarios.

\section{Methodology}
\label{sec:method}

\subsection{Marginal ratio and augmented marginal ratio estimators}

Recall that the IPW and AIPW estimators are defined as $\E\left[ h(A,X) Y \right]$ and $\E\left[ h(A,X) Y^{*} \right]$ in the population version. We denote the population weights, $w(Y) = \E\left[h(A,X)\middle | Y\right]=\E\left[\frac{A-\pi(X)}{\pi(X)\left\{1-\pi(X)\right\}} \middle | Y\right]$ and $w^{*}(Y^{*}) = \E\left[h(A,X)\middle | Y^*\right]=\E\left[\frac{A-\pi(X)}{\pi(X)\left\{1-\pi(X)\right\}} \middle | Y^{*}\right]$. These weights are the key components of MR and AMR, respectively. The corresponding estimated weights functions are denoted $\hat{w}(\cdot)$ and $\hat{w}^{*}(\cdot)$, when we use estimated nuisance parameters $\hat{\pi}(\cdot)$ and $\hat{\mu}^a(\cdot)$. With these notations, we define:

\begin{definition}[MR and AMR estimators]
Our proposed empirical MR and AMR estimators are
$$
\hat{\theta}_{MR}=\P_n\left[\hat{w}(Y)Y\right] = \P_n \left\{\hat{\E}\left[\frac{A-\hat{\pi}(X)}{\hat{\pi}(X)\left\{1-\hat{\pi}(X)\right\}} \middle | Y\right]Y\right\},
$$
and
$$
\hat{\theta}_{AMR}=\P_n\left[\hat{w}^{*}(\hat{Y}^{*})\hat{Y}^{*}\right]=\P_n \left\{\hat{\E}\left[\frac{A-\hat{\pi}(X)}{\hat{\pi}(X)\left\{1-\hat{\pi}(X)\right\}} \middle | \hat{Y}^{*}\right]\hat{Y}^{*}\right\}.
$$
\end{definition}
For later discussion, we also introduce the oracle estimators 
$$
w^0 (\cdot)= \E\left[\hat{h}(A,X)\middle|Y\right] = \E\left[\frac{A-\hat{\pi}(X)}{\hat{\pi}(X)\left\{1-\hat{\pi}(X)\right\}} \middle | Y \right]
$$ 
and 
$$
w^{*0}(\cdot) = \E\left[\hat{h}(A,X)\middle|\hat{Y}^{*}\right]=\E\left[\frac{A-\hat{\pi}(X)}{\hat{\pi}(X)\left\{1-\hat{\pi}(X)\right\}} \middle | \hat{Y}^{*} \right]
$$ as intermediate quantities.  \Cref{tab:notations} presents a summary of these definitions. The Oracle column contains $\hat{\theta}^0_{IPW}$ and  $\hat{\theta}^0_{AIPW}$ based on empirical distributions $\P_n$ while assuming the true values of nuisance parameters are known, as well as  $\hat{\theta}^0_{MR}$ and  $\hat{\theta}^0_{AMR}$ that uses the estimated nuisance parameters $\hat{\pi}, \hat{\mu}^a$ but assume the corresponding weights functions are known. The Estimator column instead shows the usual case in which the nuisance parameters and weights functions are unknown and estimated.  

\begin{table}
\begin{center}
\begin{tabular}{|r|r|r|r|}
    \hline
     & Population &  Oracle & Estimator  \\
    \hline
     \textbf{IPW} & $\theta_{IPW}=\E\left[h(A,X)Y\right]$ & $\hat{\theta}^0_{IPW} = \P_n\left[h(A,X)Y\right]$ &  $\hat{\theta}_{IPW}=\P_n\left[\hat{h}(A,X)Y\right]$\\
     \textbf{AIPW} & $\theta_{AIPW}=\E\left[h(A,X)Y^{*}\right]$ & $\hat{\theta}^0_{AIPW} = \P_n\left[h(A,X)Y^{*}\right]$ & $\hat{\theta}_{AIPW}=\P_n\left[\hat{h}(A,X)\hat{Y}^{*}\right]$\\
     \textbf{MR} & $\theta_{MR} =\E\left[w(Y)Y\right]$ & $\hat{\theta}^0_{MR} =\P_n\left[w^0(Y)Y\right]$ &  $\hat{\theta}_{MR} =\P_n\left[\hat{w}(Y)Y\right]$\\
     \textbf{AMR} & $\theta_{AMR} =\E\left[w^{*}(Y^{*})Y^{*}\right]$ & $\hat{\theta}^0_{AMR} =\P_n\left[w^{*0}(\hat{Y}^{*})\hat{Y}^{*}\right]$ & $\hat{\theta}_{AMR} =\P_n\left[\hat{w}^{*}(\hat{Y}^{*})\hat{Y}^{*}\right]$  \\
     \hline

\hline
\end{tabular}
\end{center}
\caption{Summary of estimators.}
\label{tab:notations}
\end{table} 

The construction of MR originates from the form of IPW, and AMR is inspired by AIPW. Note that ${Y}^{*} = Y-\big\{\pi(X)\mu^0(X) + \left\{1-\pi(X)\right\}\mu^1(X)\big\}$ is a combination of the observed outcome $Y$ and the estimation bias $Y-\mu^a$ from AIPW, and is different from the `residual of Y', that is, $Y-\E\left[Y \middle| X\right] = Y-\big\{\pi(X)\mu^1(X)+\left\{1-\pi(X)\right\}\mu^0(X)\big\}$. The $\mu^{*}(X)$ here is introduced to debias the MR estimator, following the same principle as in the AIPW estimator, where a correction term debiases the standard IPW estimator. A similar debiasing idea is discussed in \citet{wang2024debiased}.

The weight function $w^{*0}(\cdot)$ is defined as the conditional expectation over the estimated $\hat{Y}^{*}=Y-\hat{\mu}^0(X)\hat{\pi}(X) - \hat{\mu}^1(X)\{1-\hat{\pi}(X)\}$, which follows a different distribution than the true $Y^{*}=Y-\mu^0(X)\pi(X)-\mu^1(X)\{1-\pi(X)\}$ that defines $w^{*}(\cdot)$. In other words, $w^{*0}$ and $w^{*}$ differ because they are computed with respect to different outcomes---$\hat{Y}^{*}$ versus $Y^{*}$. When the estimated nuisance parameters $\hat{\mu}^a(\cdot)$ and $\hat{\pi}(\cdot)$ exactly match their true values $\mu^a(\cdot)$ and $\pi(\cdot)$, the two functions $w^*(\cdot)$ and $w^{*0}(\cdot)$ are identical.

In contrast, $w^0(\cdot)$ is the conditional expectation computed on $Y$, aligning it with the definition of $w(\cdot)$. However, these functions are not identical: $w^0$ involves the expectation of $\hat{h}(A,X)$, whereas $w(\cdot)$ involves $h(A,X)$. A simple example illustrating the differences among $w^{*0}(\cdot)$, $w^{*}(\cdot)$, $w^0(\cdot)$, and $w(\cdot)$ is provided in \Cref{ex:explict_weight}.

In applications using real data, the weights have to be estimated. In this paper, we undertake this by solving\footnote{Other loss functions, such as the Huber loss can also be used.}:
\begin{align}
    \hat{w}(\cdot) &= \arg\min_{f\in \mathcal{F}}\mathbb{P}_n\left[\frac{A-\hat{\pi}(X)}{\hat{\pi}(X)\left\{1-\hat{\pi}(X)\right\}} - f(Y)\right]^2, \label{eqn:loss_weights_mr}\\
    \hat{w}^{*}(\cdot) &= \arg\min_{f\in \mathcal{F}}\mathbb{P}_n\left[\frac{A-\hat{\pi}(X)}{\hat{\pi}(X)\left\{1-\hat{\pi}(X)\right\}} - f(\hat{Y}^{*})\right]^2 \label{eqn:loss_weights_amr}
\end{align}
within a function class $\mathcal{F}$. 

Given these notations, we provide our method for AMR treatment effect estimation in \Cref{alg:amr}.  It follows almost the same procedure to build the MR estimator, except for that there is no need to calculate $\hat{Y}_{-k}^{*}$, and we estimate $\hat{w}_{-k}(\cdot)$ by regressing $\hat{h}_{-k}(A,X)$ on $Y$.

\begin{algorithm}
\caption{Estimating ATE with Augmented Marginal Outcome Density Ratio}
\label{alg:amr}
\begin{algorithmic}

\Require Observations $\{X_i, Y_i, A_i\}_{i=1}^n$.
\Ensure AMR estimation for ATE, $\hat{\theta}_{AMR}$
\State Split the data into $K$ folds, denoted $\{D_k\}_{k=1}^K$.
\For{$k = 1$ to $K$}
    \State Train nuisance parameter estimation models on the training set $D_{-k}$ (all data except fold $D_k$) to obtain: $\hat{\mu}^a_{-k}(\cdot)$ and $\hat{\pi}_{-k}(\cdot)$ that takes covariates $x \in \mathcal{X}$ as input;
    \State Using the trained models on $D_{-k}$, compute for all observations in $D_{-k}$: $\hat{Y}^{*}_{-k} \coloneqq\left\{\hat{Y}^{*}_{i} = Y_i - \hat{\pi}_{-k}(X_i) \hat{\mu}_{-k}^0(X_i) - \{1 - \hat{\pi}_{-k}(X_i)\} \hat{\mu}_{-k}^1(X_i), \;i\in D_{-k}\right\}$; 
    \State Still on the training set $D_{-k}$, regress $\hat{h}_{-k}(A,X) \coloneqq \left\{\frac{A_i-\hat{\pi}_{-k}(X_i)}{\hat{\pi}_{-k}(X_i)(1-\hat{\pi}_{-k}(X_i))},\; i\in D_{-k} \right\}$ on $\hat{Y}_{-k}^{*}$ to estimate the weight function $\hat{w}^{*}_{-k}(\cdot)$ using (\ref{eqn:loss_weights_amr});
\EndFor
\State Compute the average across all $n$ observations:
$$
\hat{\theta}_{AMR} = \frac{1}{n} \sum_{i=1}^n \sum_{k=1}^K \mathbbm{1}\left\{i\in D_k\right\}\hat{w}^{*}_{-k}(\hat{Y}^*_{i,k}) \cdot \hat{Y}^{*}_{i,k},
$$
where $\hat{Y}^{*}_{i,k} = Y_i - \hat{\pi}_{-k}(X_i) \hat{\mu}_{-k}^0(X_i) - \{1 - \hat{\pi}_{-k}(X_i)\} \hat{\mu}_{-k}^1(X_i).$
\end{algorithmic}
\end{algorithm}

\begin{remark}[Cross-fitting]
In  \Cref{alg:amr} we fit $\hat{w}^*_{-k}(\cdot)$ on the same fold along with other nuisance parameter functions, $\hat{\pi}_{-k}(\cdot)$ and $\hat{\mu}_{-k}^a(\cdot), a\in\{0,1\}$. One can also split the fold $D_{-k}$ into $D_{-k}^0$ and $D_{-k}^1$, fit the nuisance parameters $\hat{\pi}_{-k}(\cdot)$ and $\hat{\mu}_{-k}^a(\cdot), a\in\{0,1\}$ on $D_{-k}^0$, and $\hat{w}^*_{-k}(\cdot)$ on $D_{-k}^1$. In the finite sample context, we recommend following \Cref{alg:amr} that uses all the data in $D_{-k}$ to estimate the weights function, since this makes the best use of potentially small amounts of data.
\end{remark}

\begin{remark}[Cross-validation for hyperparameter tuning]
We tune hyperparameters via cross-validation when estimating the weight functions $\hat{w}$ and $\hat{w}^{*}$. This tuning is crucial for performance gains, as it balances the bias-variance trade-off in finite samples. For example, if a regularized kernel method (such as kernel ridge regression) is used for weight estimation, the estimator incurs a bias term of order $O(\lambda + h^2)$, 
where $\lambda$ is the regularization parameter and $h$ denotes the kernel bandwidth. Under standard smoothness assumptions in nonparametric regression, the optimal bandwidth is typically of order $h \asymp n^{-1/5}$, so that $h^2 \asymp n^{-2/5}$. It is common to choose the regularization parameter $\lambda$ so that it decreases at a rate ensuring optimal convergence; see, e.g.~\citet{caponnetto2007optimal} and \citet{smale2007learning} for precise results under various assumptions. With cross-validation selecting $\lambda$ and $h$ appropriately, the induced bias diminishes as $n$ increases, ensuring that the estimator's bias vanishes in the limit.
\end{remark}

When all the nuisance parameters and the weights function are known, AMR and MR demonstrate benefits in reduced variance when compared to their counterparts AIPW and IPW, as we now show.

\begin{proposition}[Variance comparison with IPW and AIPW]
\label{prop:variance_comparison_oracle}
When the weights $w$, $w^{*}$ and the nuisance parameters $\pi, \mu^a$,  are known exactly, we have
\begin{align*}
\Var\hat{\theta}^0_{IPW}-\Var\hat{\theta}^0_{MR}  &= \frac{1}{n}\E\{\Var\left[h(A,X)\middle| Y\right]Y^2\} \geq 0\\
\Var\hat{\theta}^0_{AIPW}-\Var\hat{\theta}^0_{AMR}  &= \frac{1}{n}\E\{\Var\left[h(A,X)\middle| Y^{*}\right]{Y^{*}}^2\} \geq 0.
\end{align*}
\end{proposition}

In other words, if there is any variation in $h(A,X)$ and $Y$, MR and AMR are strictly more efficient than IPW and AIPW, respectively.  We provide the details of the proof in \Cref{sec:oracle_var_proof} which is similar to \citet{taufiq2024marginal}. The second equation in \Cref{prop:variance_comparison_oracle} shows that, if ${Y}^{*}$ is small on average, then the difference between $\Var \hat{\theta}_{AIPW}$ and $\Var \hat{\theta}^0_{AMR}$ is also small.  If the variance of $h(A,X)$ conditional on ${Y}^{*}$ is big, we get more benefit by using $\hat{\theta}^0_{AMR}$ to achieve lower variance.

\begin{proposition}[Variance comparison between MR and AMR]
\label{prop:variance_comparison_proposed}
Under the assumptions of \Cref{prop:variance_comparison_oracle}, we have 
\small
$$\Var\hat{\theta}^0_{MR} - \Var\hat{\theta}^0_{AMR} \geq \frac{1}{n}\E\left\{ \Bigl[h(A,X)^2 - \max\Bigl(\Var\left[h(A,X)\middle |{Y}^{*}\right], \Var\left[h(A,X)\middle| Y\right]\Bigr) \Bigr](Y^2 - {Y^{*}}^2)\right\}.$$
\end{proposition}
The proof can be found in \Cref{sec:var_proof}. The term on the right is not necessarily greater than 0. However, when the treatment effect is homogeneous, i.e.~$\mu^1(x) = \mu^0(x), \forall x\in \mathcal{X}$, then we have $Y^* = Y-\E Y$. In this case, $\E[Y^2 - {Y^{*}}^2\mid X] = 0$, thus we get $\Var\hat{\theta}^0_{MR} \geq \Var\hat{\theta}^0_{AMR}$.

So far we have demonstrated the advantages of the proposed $\hat{\theta}_{MR}$, $\hat{\theta}_{AMR}$ estimators when true values of nuisance parameters are known. However, $\mu^a(x)$ and $\pi(x)$ are usually unknown and need to be estimated  in practice. The strengths of $\hat{\theta}_{AMR}$, which are inherited from $\hat{\theta}_{AIPW}$, further stand out in such settings---we show that it maintains the double robustness and asymptotic normality properties in \Cref{sec:asymptotic}. Before we present the asymptotic properties, we first present the intuition and strengths of our proposed estimators.

\subsection{Outcome-informative weighting}
\label{sec:outcome_driven_weighting}

The most straightforward interpretation of our weighting approach is that it assigns the same weight to all observations sharing identical values of $Y$ or $Y^*$.  
This is particularly helpful not just because it smooths the weights as we will discuss later in \Cref{subsec:connect_efficient_set}, but also when $Y$ or $Y^*$ is informative for average treatment effect estimation. We explain this idea with \Cref{ex:explict_weight}. The derivation is detailed in \Cref{sec:true_weights_derivation}.

\begin{example}[Explicit form of MR and AMR]
\label{ex:explict_weight}
    Assume $A|X \sim \textit{Bernoulli}(\pi(X))$ and $Y^a\mid X\sim \mathcal{N}(\mu^a(X), \sigma^2)$ for  $a\in\{0,1\}$. Write $\mathcal{N}(z;\;a,\;b)$ for the density of $z$ under a Gaussian distribution with mean $a$ and variance $b$. We obtain:
\begin{align*}
    w(y) &= \frac{\E_X\Bigl[\mathcal{N}\big(y;\; 
 \mu^1(X),\;\sigma^2\big) - \mathcal{N}\big(y;\; \mu^0(X),\;\sigma^2\big)\Bigr]}{ \E_X\Bigl[\pi(X)\mathcal{N}\big(y;\;\mu^1(X),\;\sigma^2\big) + \left\{1-\pi(X)\right\}\mathcal{N}\big(y;\;\mu^0(X),\;\sigma^2\big)\Bigr]}\\
    w^{*}(y^{*}) &=  \frac{\E_X\Bigl[\mathcal{N}\big(y^{*};\;\pi(X)\tau(X),\;\sigma^2\big) - \mathcal{N}\big(y^{*};\; -\left\{1-\pi(X)\right\}\tau(X),\;\sigma^2\big)\Bigr]}{ \E_X\Bigl[\pi(X)\mathcal{N}\big(y^{*};\;\pi(X)\tau(X),\;\sigma^2\big)+ \left\{1-\pi(X)\right\}\mathcal{N}\big(y^{*};\;-\left\{1-\pi(X)\right\}\tau(X),\;\sigma^2\big)\Bigr]}
\end{align*}
where $\tau(x)\coloneqq \mu^1(x) - \mu^0(x)$.

Meanwhile, we can also derive the weights with plug-in nuisance estimators:
\begin{align*}
    w^0(y) &= \frac{\E_X\Bigl[\frac{\pi(X)}{\hat{\pi}(X)}\mathcal{N}\big(y;\; \mu^1(X),\;\sigma^2\big) - \frac{1-\pi(X)}{1-\hat{\pi}(X)}\mathcal{N}\big(y;\;\mu^0(X),\sigma^2\big)\Bigr]}{ \E_X\Bigl[\pi(X)\mathcal{N}\big(y;\;\mu^1(X),\;\sigma^2\big) + \left\{1-\pi(X)\right\}\mathcal{N}\big(y;\;\mu^0(X),\;\sigma^2\big)\Bigr]}\\
    w^{*0}(y^{*}) &= \frac{\E_X\Bigl[\frac{\pi(X)}{\hat{\pi}(X)}\mathcal{N}\big(y^{*};\;\mu^1(X) - \hat{\mu}^{*}(X),\;\sigma^2\big) - \frac{1-\pi(X)}{1-\hat{\pi}(X)}\mathcal{N}\big(y^{*};\;\mu^0(X) - \hat{\mu}^{*}(X),\;\sigma^2\big)\Bigr]}{ \E_X\Bigl[\pi(X)\mathcal{N}\big(y^{*};\;\mu^1(X) - \hat{\mu}^{*}(X),\;\sigma^2\big)+ \left\{1-\pi(X)\right\}\mathcal{N}\big(y^{*};\;\mu^0(X) - \hat{\mu}^{*}(X),\;\sigma^2\big)\Bigr]}.
\end{align*}
Clearly, they are different from $w(y)$ and $w^{*}(y^*)$, respectively.
\end{example}

In \Cref{ex:explict_weight}, when there is no treatment effect, i.e.~$\forall x\in \mathcal{X}$, we have $\mu^1(x)=\mu^0(x)$, or $\tau(x)=0$, our proposed weights correctly become $0$. This reflects how the approach effectively averages out irrelevant variation in $h(A,X)$, preventing extreme values from dominating. Consequently, the weights are distributed in a manner that is closer to being proportionate to the true relationship with $Y$ or $Y^{*}$, mitigating the influence of outliers.  This highlights the critical role that the outcome information plays in the weights design. 

In \Cref{fig:example_weights} we provide an example where the covariates and the outcome generating process follow the setup in \Cref{ex:explict_weight}, and we set $\tau(x)=0, \forall x \in \mathcal{X}$. We estimate the weight function $\hat{w}^{*}(y^*)$ using linear regression, since we know the true weights are  $w^{*}(y^*)=0$ for $y\in \mathbb{R}$. Following the imbalance measure described in \cite{imai2014covariate}, the imbalance metric for a covariate $X_j$, $j\in \{1,\ldots, p\}$ is defined as 
  $$
  \text{Imbalance}(X_j) = \Bigl\{\P_n\left[b(X)X_j)\right]^\top\P_n\left[(X_j^\top X_j)\right]^{-1}\P_n\left[b(X)X_j)\right]\Bigr\}^{-1/2},
  $$ 
where $b(x_i)$ defines the weight applied on observation $i$; for example, IPW uses weights $b(x_i)=\frac{a_i-\pi(x_i)}{\pi(x_i)\left\{1-\pi(x_i)\right\}} $. We present the imbalance metric without standardizing by the weights to emphasize its scale. Since the covariates do not influence the treatment effect, and the treatment effect is set to 0, we expect that all individuals receive an equal weight of $0$ and that the covariates remain balanced between the treatment arms.  \Cref{fig:example_weights} shows that the AMR method yields these ideal weights even with a relatively small sample size ($n=100$), whereas the IPW approach only gradually reaches the correct weight scale as the sample size increases and, in some cases, even exacerbates the imbalance of certain covariates.
\begin{figure}[ht]
    \centering
    \begin{subfigure}[b]{0.45\textwidth}
        \centering
\includegraphics[width=1\linewidth]{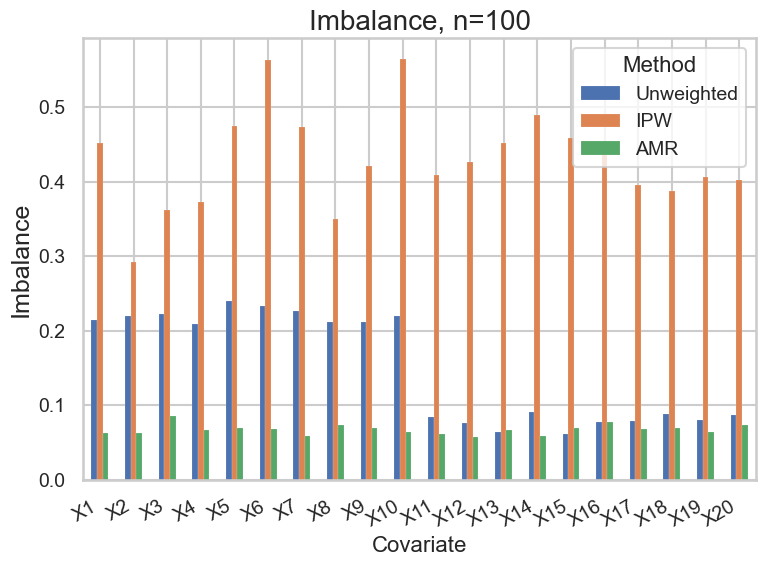}

    \end{subfigure}
    \hfill
    \begin{subfigure}[b]{0.45\textwidth}
        \centering
\includegraphics[width=1\linewidth]{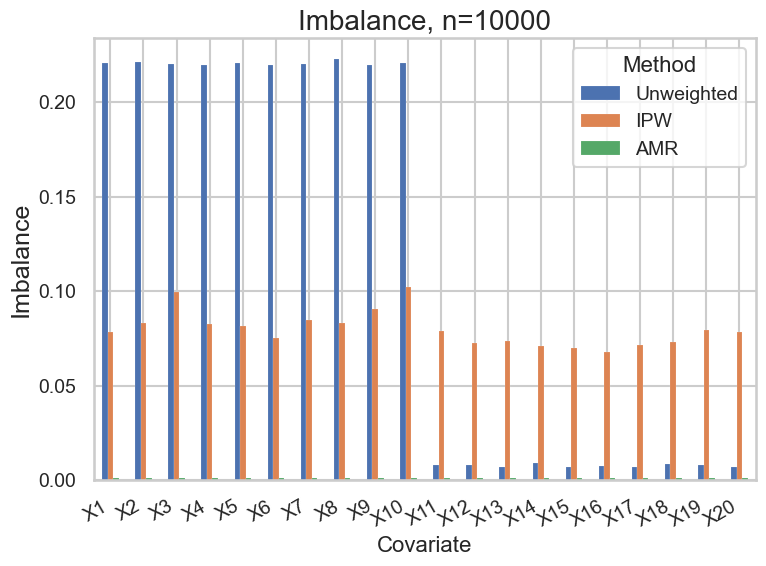}
    \end{subfigure}
    \caption{Imbalance of covariates under different weights design when $\tau(x) = 0, \forall x\in \mathcal{X}$. $(X_1,\ldots, X_5)$ are instrumental variables; $(X_6,\ldots, X_{10})$ represent confounders; $(X_{11},\ldots, X_{15})$ stand for prognostic variables and $(X_{16},\ldots, X_{20})$ are spurious covariates. }
    \label{fig:example_weights}
\end{figure}

\subsection{Connection with estimation methods based on covariate adjustment}
\label{subsec:connect_efficient_set}

MR and AMR use weights that are outcome informative, but here we demonstrate their connection with covariate adjustment related methods. We give two examples in \Cref{fig:example}, where $I$ represents the instrumental variables---those are correlated with the treatment $A$ and affecting the outcome $Y$ only through treatment; $O$ contains the prognostic variables---those that predict the outcome without being related to the treatment; $C$ represents the confounding variables that are correlated with both the treatment and the outcome; $S$ consists of spurious variables---those representing random noise with no direct correlation with either the outcome or treatment. In what follows, we focus on the MR weights.

\begin{figure}[ht]
    \centering
    \begin{subfigure}[b]{0.45\textwidth}
        \centering
        \begin{tikzpicture}[node distance=2cm, >=stealth', auto]

            \node (XI) {$I$};
            \node (Xo) [right of=XI] {$O$};
            \node (Xs) [right of=Xo, xshift=-0.5cm, yshift=-1cm] {$S$};

            \node (A) [below of=XI, yshift=-1cm] {$A$};
            \node (Y) [below of=Xo, yshift=-1cm] {$Y$};

            \draw[->] (XI) -- (A);
            \draw[->] (A) -- (Y);
            \draw[->] (Xo) -- (Y);
        \end{tikzpicture}
        \caption{Without confounder}
        \label{fig:wo_confounder}
    \end{subfigure}
    \hfill
    \begin{subfigure}[b]{0.45\textwidth}
        \centering
        \begin{tikzpicture}[node distance=2cm, >=stealth', auto]

            \node (XI) {$I$};
            \node (XC) [right of=XI] {$C$};
            \node (Xo) [right of=XC] {$O$};
            
            \node (Xs) [right of=Xo, xshift=-0.5cm, yshift=-1cm] {$S$};

            \node (A) [below of=XI, xshift=1cm, yshift=-1cm] {$A$};
            \node (Y) [below of=Xo, xshift=-1cm, yshift=-1cm] {$Y$};
            
            \draw[->] (XI) -- (A);
            \draw[->] (XC) -- (A);
            \draw[->] (XC) -- (Y);
            \draw[->] (Xo) -- (Y);
            \draw[->] (A) -- (Y);
        \end{tikzpicture}
        \caption{With confounder}
        \label{fig:w_confounder}
    \end{subfigure}
    \caption{Example causal DAGs.}
    \label{fig:example}
\end{figure}
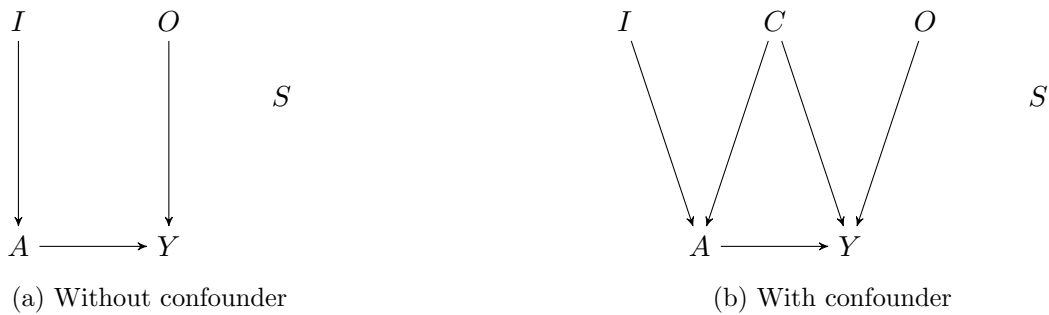

\begin{proposition}[MR Weights without Confounders]
\label{prop:mr-weights-no-confounders}
Suppose $X=\{I,O,S\}$, i.e., there is no confounder, as shown in \Cref{fig:wo_confounder}. Since $O$ is prognostic variable, we know $\pi(X) = \pi(I)$. We have $Y \ind I \mid A$. Then the MR weight satisfies
\begin{align*}
    \E\left[\frac{A-\pi(X)}{\pi(X)\left\{1-\pi(X)\right\}} \middle| Y\right] &= \sum_{a\in\{0,1\}}\P(A=a\mid Y)\times\E\left[\frac{A-\pi(I)}{\pi(I)\{1-\pi(I)\}}\middle| Y,A=a\right]\\
    &=\P(A=1\mid Y)\E\pi(I)^{-1} - \P(A=0\mid Y)\E\{1-\pi(I)\}^{-1} .
\end{align*}
\end{proposition}
 \Cref{prop:mr-weights-no-confounders} shows that the MR weights when there is no confounder do not depend on the individual propensity score as desired. When confounding exists, the situation becomes more complex:
\begin{proposition}[MR Weights with Confounders]
\label{prop:mr-weights-confounders}
Let $X = \{I, C, O, S\}$ as in \Cref{fig:w_confounder}, where $I \ind Y \mid A,C$. We have $\pi(X) = \pi(I,C)$. Then conditioning on $Y$ removes the influence of $I$ from the propensity score adjustment because:
\begin{align*}
    \E\left[\frac{A-\pi(X)}{\pi(X)\left\{1-\pi(X)\right\}} \middle| Y\right]
    &=\sum_{a\in\{0,1\}}\int \P(A=a\mid Y)p_{C\mid AY}(c\mid Y,A=a)\\
    &\qquad\qquad\times \E\left[\frac{A-\pi(I,C)}{\pi(I,C)\{1-\pi(I,C)\}}\middle| Y,A=a, C=c\right] \mathrm{d}c \\
    &=\sum_{a\in\{0,1\}}\int \P(A=a\mid Y)p_{C\mid AY}(c\mid Y,A=a)\E\left[\frac{a-\pi(I,c)}{\pi(I,c)\{1-\pi(I,c)\}}\right]\mathrm{d}c\\
    &= \P(A=1\mid Y)\int p_{C\mid AY}(c\mid Y,A=1)\E\pi(I,c)^{-1}\mathrm{d}c\\
    &\qquad\qquad-\P(A=0\mid Y)\int p_{C\mid AY}(c\mid Y,A=0)\E\{1-\pi(I,c)\}^{-1}\mathrm{d}c.
\end{align*} 
\end{proposition}
The last equality in \Cref{prop:mr-weights-confounders} relies on the assumption $I \ind Y \mid A,C$. Under this conditional independence, the expectation term $\E\left[\pi(I,c)^{-1}\{1-\pi(I,c)\}^{-1}\left\{a-\pi(I,c)\right\}\right]\,$ is independent of $Y$, and thus the variation in the instrumental variable $I$ is averaged out. Consequently, the instrumental variable $I$ does not affect the MR weights once we condition on $Y$, as its influence remains constant across different values of $Y$.

By using conditional expectation, both the AMR and MR weights effectively filter useful information from all variables, achieving a more efficient aggregation than stepwise variable selection. \Cref{fig:projection} provides a geometric view of this information filtering step. The conditional expectation,  $\E\left[h\left(A,X\right) \middle| Y \right]$,  is the orthogonal projection of $h(A,X)$ onto the closed subspace of all random variables that are measurable with respect to the $\sigma$-algebra generated by $Y$ (denoted $\sigma(Y)$). The weight $w$ in the MR estimator presents the component containing all the information about  $h(A,X)$ that can be inferred from $Y$, discarding the orthogonal (uninformative) component. In particular, the term $h(A,X) - \E\left[h(A,X)\middle|Y\right]$ is orthogonal to $\sigma(Y)$, so it does not contribute to the estimation $\E\left[h(A,X)Y\right]$. Formally, 
$\E\Big[\big\{h(A,X) - \E\left[h(A,X)\middle|Y\right]\big\}Y\Big] = 0$.   Excluding this orthogonal term effectively removes the noise that is uninformative for $Y$, reducing the variance without introducing bias. This highlights the difference between our work and the e-score---we present a short discussion and two alternatives in \Cref{sec:e-score} for interested readers.

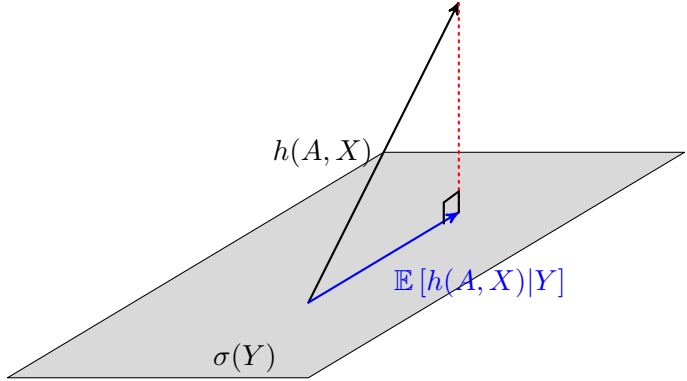
\begin{figure}[ht]
\centering
\begin{tikzpicture}[scale=2, line cap=round, line join=round, >=stealth']

    \coordinate (O) at (0,0);             
    \coordinate (vo) at (2, 0.5); 
    \coordinate (A) at (2.,0);         
    \coordinate (B) at (2.5,1.5);         

    \fill[gray!30] (O) -- (A) -- ($(A)+(B)$) -- ($(O)+(B)$) -- cycle;
    \draw (O) -- (A) -- ($(A)+(B)$) -- ($(O)+(B)$) -- cycle;

    \node[below right] at ($(O)!0.5!(A)!0.2!(B)$) {$\sigma(Y)$};

    \coordinate (H) at (3.,2.5); 
    \draw[->, thick] (vo) -- (H)
        node[midway, left] {\(h(A,X)\)};

    \coordinate (P) at (3,1.1); 
    \draw[->, thick, blue] (vo) -- (P)
        node[midway, below right, text=blue] {$\E\left[h(A,X)\middle| Y\right] $};

    \draw[dotted, thick, red] (P) -- (H);

    \coordinate (marker_start) at ($(P)!0.1!(H)$); 
    \coordinate (marker_horiz) at ($(marker_start)!0.1!(vo)$); 
    \coordinate (marker_vert) at ($(marker_start)!1!(P)$);  

    \draw[thick] (marker_start) -- (marker_horiz);
    \draw[thick] (marker_start) -- (marker_vert);
    \draw[thick] (marker_horiz) -- ($ (marker_horiz) + (marker_vert) - (marker_start) $);

\end{tikzpicture}
\caption{Projection of $h(A,X)$ onto $\sigma(Y)$.}
\label{fig:projection}
\end{figure}

Despite the similarities shared with methods directly targeted at covariate adjustment, we would like to emphasize that the proposed MR and AMR estimators build on outcome adjustment, which is different from previous re-weighting or balancing methods that focus on the covariate space.  Methods like the e-score and DOPE, although they point to the importance of adjusting propensity score estimation via estimated outcome mean (or its bias), still project all information to the covariate space (e.g.~e-score involves regressing the outcome mean estimation bias on covariates $X$), losing information from $Y$. Introducing another layer of such model assumption for projection on $X$ also requires more decision steps, resulting in higher risk of model misspecification. Our proposed estimators extract information without imposing structural assumptions. Viewing the conditional expectation as a \emph{post-hoc} adjustment allows us to flexibly model the nuisance parameters and circumvent the propagation of uncertainty and bias introduced by \emph{a priori} assumptions. With that being said, $w^{*}(Y^{*})$ is not necessarily an ``efficient adjustment set'' as discussed in \citet{christgau2024efficient}, because we shift the focus to the outcome space $Y$ or $Y^*$. We will see the difference more clearly in $\hat{\theta}_{AMR}$'s asymptotic variance provided in \Cref{sec:asymptotic}.

\section{Double robustness and asymptotic normality}\label{sec:asymptotic}
In this section we discuss the double robustness and the asymptotic normality of $\hat{\theta}_{AMR}$. 
\subsection{Asymptotic properties}\label{subsec:asymptotic}

\begin{theorem}[Double robustness of $\hat{\theta}_{AMR}$]
\label{thm:DR-AMR}
Suppose the following are satisfied:
\begin{itemize}

\item \textbf{Boundedness:} 

$\pi(x)^{-1}$ and
$\mu^a(x)$, $a\in\{0,1\}$ along with their corresponded estimators are bounded; $Y$, $Y^{*}$, and $\hat{Y}^{*}$ are square-integrable;

\item \textbf{Consistency of $\hat{w}^{*}(\cdot)$}: $ \lVert\hat{w}^{*}-w^{*0}\rVert_2 = \operatorname{o}_\P(1)$.
\end{itemize}

Then, if either $\lVert \hat{\pi} - \pi \rVert_2 \;\overset{P}{\longrightarrow}\;0$, or \ $\lVert \hat{\mu}^a-\mu^a\rVert_2 \;\overset{P}{\longrightarrow}\;0,\; \text{for both} \; a \in \{0,1\}$, we have 
$$
\hat{\theta}_{AMR}\;\overset{P}{\longrightarrow}\; \theta.
$$
\end{theorem}
The proof is provided in \Cref{proof:DR-AMR}. 

Since $w^{*0}$ is a one-dimensional regression function,  achieving $ \lVert\hat{w}^{*}-w^{*0}\rVert_2 = \operatorname{o}_\P(1)$ is relatively straightforward using methods like parametric or kernel smoothing. This represents a notable advantage over preprocessing strategies like representation learning or variable selection techniques (e.g., \citealt{benkeser2020nonparametric, gui2022causal}) that require $\hat{\mu}^a$ to be consistent,  as well as approaches like \citealt{christgau2024efficient}, which hinges on estimating $\beta$ in (\ref{eqn:single-index}). The latter methods all involve higher-dimensional estimation, making them comparatively more challenging than our one-dimensional setting.

The proofs are constructed based on the relationships between AIPW and AMR estimators, as illustrated in \Cref{fig:relations}. The crucial link is that $\hat{\theta}^0_{AMR} = \P_n\left[w^{*0}(\hat{Y}^*)\hat{Y}^*\right]$ serves as the empirical mean estimator for $\E\hat{\theta}^0_{AMR}$ , which is identical to $\E\hat{\theta}_{AIPW}$. This connection links $\hat{\theta}_{AIPW}$ to $\theta$ in terms of consistency.

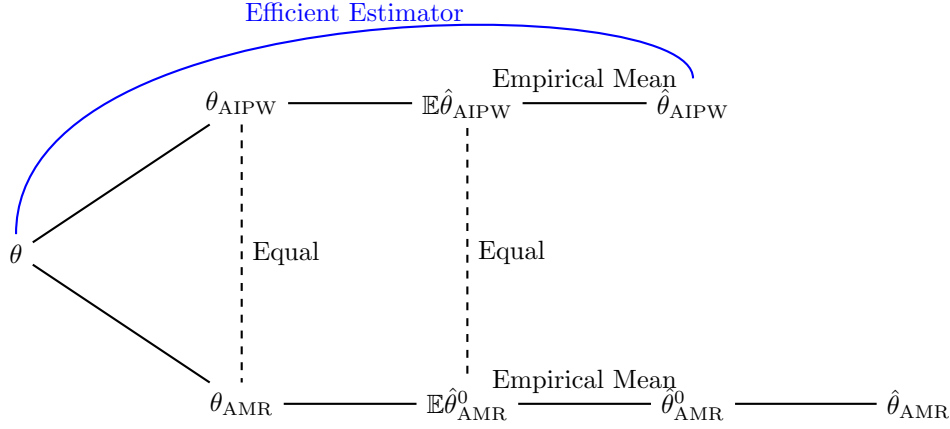
\begin{figure}[ht]
    \centering
    \begin{tikzpicture}[
        every node/.style={font=\small, align=center},
        line/.style={thick, -},
        label/.style={font=\small, midway, above}
    ]
        \node (theta) at (0,0) {$\theta$};
        
        \node (theta_aipw) at (3,2) {$\theta_{\text{AIPW}}$};
        \node (E_aipw)     at (6,2) {$\mathbb{E}\hat{\theta}_{\text{AIPW}}$};
        \node (hat_aipw)   at (9,2) {$\hat{\theta}_{\text{AIPW}}$};
        
        \node (theta_amr) at (3,-2) {$\theta_{\text{AMR}}$};
        \node (E_amr)     at (6,-2) {$\mathbb{E}\hat{\theta}^0_{\text{AMR}}$};
        \node (hat_amr_oracle)   at (9,-2) {$\hat{\theta}^0_{\text{AMR}}$};
        \node (hat_amr)   at (12,-2) {$\hat{\theta}_{\text{AMR}}$};
        
        \draw[line] (theta) -- (theta_aipw);
        \draw[line] (theta_aipw) -- (E_aipw);
        \draw[line] (E_aipw) -- node[label]{Empirical Mean} (hat_aipw);
        
        \draw[line] (theta) -- (theta_amr);
        \draw[line] (theta_amr) -- (E_amr);
        \draw[line] (E_amr) -- node[label]{Empirical Mean} (hat_amr_oracle);
        \draw[line] (hat_amr_oracle) -- (hat_amr);
        
        \draw[line, dashed] (theta_aipw) -- node[midway, right]{Equal} (theta_amr);
        \draw[line, dashed] (E_aipw) -- node[midway, right]{Equal} (E_amr);
        
        \draw[line, thick, blue] (hat_aipw.north) 
            .. controls (9,3.5) and (0,3.5) .. (theta.north)
            node[midway, above, font=\small]{Efficient Estimator};
    \end{tikzpicture}
\caption{Relationships of  AIPW and AMR estimators.}
    \label{fig:relations}
\end{figure}
We now focus on establishing asymptotic normality. For simplicity, we denote (\ref{eqn:eif}) as $\varphi_{AIPW}(Z;\theta,\xi)$. With this notation, the estimating function for $\theta_{AMR}$ can be expressed as
\begin{equation*}
\label{eif_amr}
\varphi_{AMR}\big(Z;\theta, \xi\big) = w^{*}\left(Y^{*}\right)Y^*-\theta.
\end{equation*}
By definition, we have $\E\varphi_{AMR}\big(Z;\theta, \xi\big)=0$. 

Next, consider the estimating function for $\hat{\theta}^0_{AMR}$:
\begin{equation*}
\label{eqn:eif_amr_oracle} 
 \varphi_{AMR}\left(Z;\theta,\hat{\xi}\right) = w^{*0}(\hat{Y}^{*})\hat{Y}^{*}-\theta.
\end{equation*}
The estimator $\hat{\theta}^0_{AMR}$ is actually defined as the solution to $\P_n\varphi_{AMR}\big(Z;\hat{\theta}^0_{AMR},  \hat{\xi}\big) = 0$. With these expressions in hand, we now turn our attention to establishing the asymptotic normality of $\hat{\theta}^0_{AMR}$. 

\begin{theorem}[Asymptotic normality of $\hat{\theta}^0_{AMR}$]\label{thm:asymp_normal}
Assume that \Cref{ass:common} and the conditions in \Cref{thm:DR-AMR} hold, with the other conditions required for $\hat{\theta}_{AIPW}$ to achieve asymptotic normality are also satisfied:
\begin{enumerate}
    \item $\lVert \hat{\pi}-\pi \rVert_2 \lVert \hat{\mu}^a -\mu^a \rVert_2 \;=\; \operatorname{o}_\P(n^{-1/2}),$
    \item $\lVert\hat{\pi} -\pi\rVert_2 \overset{P}{\longrightarrow}0 \text{ and } \lVert\hat{\mu}^a -\mu\rVert_2 \overset{P}{\longrightarrow}0$
\end{enumerate}

for $a\in \{0,1\}$. Under these assumptions, the estimator  $\hat{\theta}^0_{AMR}$ satisfies:
$$
\sqrt{n}\left(\hat{\theta}^0_{AMR} - \theta\right)  \;\rightsquigarrow\; \mathcal{N}\left(0,\sigma^2_{AMR,0}\right)
$$
where $\sigma^2_{AMR,0} = \Var \left\{\varphi_{AMR}\left(Z;\theta,\xi\right)\right\}$.
\end{theorem}
By comparing $\varphi_{AIPW}\left(Z;\theta,\xi\right)$ and $\varphi_{AMR}\left(Z;\theta,\xi\right)$, we arrive at the following result:
\begin{proposition}[Efficiency bound comparison]
    \label{prop:eb_comparison}
    
    $$\Var\left\{\varphi_{AIPW}\left(Z;\theta,\xi\right)\right\} - \Var\left\{\varphi_{AMR}(Z;\theta,\xi)\right\} = \Var\left\{h(A,X)Y^*-w^*(Y^*)Y^*\right\}\geq 0.$$
\end{proposition}

The proof of \Cref{prop:eb_comparison} can be found in the Appendix. It implies that $\hat{\theta}^0_{AMR}$ is more asymptotically efficient than $\hat{\theta}_{AIPW}$. At first glance, this result might be surprising---why would our estimator surpass the AIPW estimator? Yet, as noted in \Cref{subsec:related_work}, employing an efficient adjustment set or representation can indeed boost estimation efficiency. A key goal of our approach is to extract and retain only the most informative features/representation, and this information filtering in AMR, which parallels the intuition explained in \Cref{subsec:connect_efficient_set}, leads to a reduced efficiency bound.

A consistent estimator for the variance $\sigma_{AMR,0}^2$ is
$$\hat{\sigma}^2_{AMR,0} = \P_n\left(w^{*0}(\hat{Y}^*)\hat{Y}^*-\hat{\theta}^0_{AMR}\right)^2
 $$ provided that $w^{*0}$ is available.

In practice, however, $w^{*0}$ is typically unknown and must be estimated. Consequently, to extend the asymptotic normality result to the estimator $\hat{\theta}_{AMR}$ (which uses an estimated version of $w^{*0}$), a more stringent assumption is required. 

\begin{theorem}[Asymptotic normality of $\hat{\theta}_{AMR}$]\label{thm:asymp_normal_strict}Under all the conditions specified in \Cref{thm:asymp_normal} and an extra requirement,
\begin{equation}
    \lVert \hat{w}^{*} - w^{*0}   \rVert_2 = \operatorname{o}_\P(n^{-1/2}),
\label{eqn:strict_convergence}
\end{equation}
the estimator $\hat{\theta}_{AMR}$ satisfies
$$
\sqrt{n}\left(\hat{\theta}_{AMR} - \theta\right)  \;\rightsquigarrow\; \mathcal{N}\left(0,\sigma_{AMR,0}^2\right).
$$
\end{theorem}
When (\ref{eqn:strict_convergence}) is met, we can construct the $(1-\alpha)$ Wald-type confidence interval as 
$$
\left[\hat{\theta}_{AMR} - z_{1-\alpha/2}\frac{\hat{\sigma}_{AMR}}{\sqrt{n}}, \hat{\theta}_{AMR} + z_{1-\alpha/2}\frac{\hat{\sigma}_{AMR}}{\sqrt{n}}\right], 
    \label{eqn:confidence_interval}
$$ where $\hat{\sigma}_{AMR}^2 = \P_n\left(\hat{w}^{*}(\hat{Y}^{*})\hat{Y}^{*} -\hat{\theta}_{AMR}\right)^2$.

However, the convergence condition on (\ref{eqn:strict_convergence}) is very strict and holds only in special cases---for instance, when $w^{*0}(\cdot)$ lies in the reproducing kernel Hilbert space (RKHS) and is infinitely smooth. 

Although the strict condition is rarely satisfied in practice, we still would like to give a reasonably useful confidence interval. We can instead construct a conservative confidence interval using the AIPW estimator. The following proposition provides such an interval:

\begin{proposition}[Conservative confidence interval for $\hat{\theta}_{AMR}$]\label{prop:conservative_ci} Assuming the conditions in \Cref{thm:asymp_normal} hold and, in addition, for every $\epsilon>0$, there exists an integer $n_0$ that $\forall n \geq n_0$,
\begin{equation}
    \P\left(\lvert\hat{\theta}_{AMR}-\hat{\theta}^0_{AMR}\rvert > z_{1-\alpha/2}\cdot\Delta\cdot n^{-1/2}\right)<\epsilon, 
\label{eqn:error_consistency}
\end{equation}
where $$\Delta\coloneqq \Var\left\{\varphi_{AIPW}(Z;\theta,\xi))\right\}- \Var\left\{\varphi_{AMR}(Z;\theta,\xi)\right\}.$$

Then, one may use the confidence interval
\begin{equation}
\label{eqn:ci_conservative}
\left[\hat{\theta}_{AMR} - z_{1-\alpha/2}\frac{\hat{\sigma}'_{AMR}}{\sqrt{n}}, \hat{\theta}_{AMR} + z_{1-\alpha/2}\frac{\hat{\sigma}'_{AMR}}{\sqrt{n}}\right],
\end{equation} 
with $\left(\hat{\sigma}'_{AMR}\right)^2=\P_n\left(\hat{h}(A,X)\hat{Y}^* -\hat{\theta}_{AMR}\right)^2$. This interval is asymptotically conservative, meaning its nominal coverage is at least $1-\alpha$.
\end{proposition}

As the sample size increases, the width of the conservative confidence interval (\ref{eqn:ci_conservative}) converges to the same width as that of the AIPW estimator, making it a reasonable choice in practice.  Although they have the same length, the confidence interval is centered around AMR estimator rather than that of AIPW. As AIPW can be very biased in finite samples, its confidence interval will often undercover the true value; this is mitigated by using AMR instead.  Additionally, as $h(A,X)$ varies less, the efficiency bound difference $\Delta=0$ gets smaller and closer to zero. Under these conditions, the advantage of variance reduction offered by AMR decreases. In the particular case where $\Delta=0$, which implies that $ \Var\left\{h(A,X)Y^*-w^*(Y^*)Y^*\right\}=0$, both AIPW and AMR achieve the same efficiency. Since the efficiency bounds are identical, it is not worthwhile risking the potential estimation error introduced by estimating $w^{*}$. Therefore, AIPW is the preferable choice in this case.

In summary, while the condition (\ref{eqn:strict_convergence}) enables the construction of an efficient, asymptotically normal estimator $\hat{\theta}_{AMR}$ with a standard confidence interval, it is a strong assumption that is rarely met. The conservative confidence interval, which relies on the weaker error bound (\ref{eqn:error_consistency}) provides a practical alternative, particularly in situations with significant lack of overlap. Our empirical findings (see \Cref{subsec:synthetic_experiment}) show that the conservative confidence intervals constructed as in (\ref{eqn:ci_conservative}) indeed tend to yield larger coverage of the true parameter.

The assumptions for asymptotic normality and the efficiency bound for $\hat{\theta}^0_{AMR}$ are different from those in (\ref{eqn:adjusted_eff_inf}). Covariate adjustment methods \citep[such as][]{christgau2024efficient} condition on various covariates, while AMR and MR adjust weights in the outcome space. Although choosing the right covariates might lower the efficiency bound, it often depends on strong assumptions and a robust representation, which can be hard to achieve. In contrast, the efficiency gain in AMR over AIPW only requires that $\Var\left[h(A,X)\middle|Y^*\right]>0$ almost surely. As a post-hoc calibration method, AMR avoids the complex assumptions needed for covariate adjustment and is less sensitive to errors in covariate selection, especially when the covariates are high-dimensional. We demonstrate this in later experiments.

\begin{remark}[Confidence intervals of $\hat{\theta}_{AIPW}$ in the context of insufficient overlap]
In our experiments, when the overlap is weak, we observe that the distribution of $\hat{\theta}_{\text{AIPW}}$ has heavy tails, as shown in \Cref{fig:synthetic_density_compare}. This aligns with the observation that although the confidence intervals of $\hat{\theta}_{\text{AIPW}}$ can achieve the nominal coverage, its MSE can be much larger than those of $\hat{\theta}_{\text{AMR}}$ and $\hat{\theta}_{\text{MR}}$. Thus, under insufficient overlap in finite samples, $\hat{\theta}_{\text{AIPW}}$ can significantly deviate from the true value of $\theta$, underscoring its limitations in such scenarios.
\end{remark}

\subsection{Weights estimation}\label{sec:weights_estimation}
In practice, functions $w^{*0}$ and $w^0$ are typically unknown in practice and need to be estimated, as in (\ref{eqn:loss_weights_mr}) and (\ref{eqn:loss_weights_amr}). In finite samples, the weights estimation bears the responsibility of balancing the bias-variance trade-off; asymptotically, as we show in \Cref{subsec:asymptotic},  the  weights estimation needs to be consistent to achieve double robustness. 

We have explored various model classes for weight estimation, including kernel smoothing (Nadaraya-Watson) regression, multilayer perceptrons (MLP), and debiased nonparametric kernel-based local polynomial regression \citep{calonico2018effect}. Although \citet{taufiq2024marginal} report favorable performance using MLP, this approach is less compelling in our setting because only one covariate---$Y$ or $\hat{Y}^{*}$---is available for weight estimation. As illustrated in \Cref{fig:example_weights} and \Cref{ex:explict_weight} with $\tau(x) = 0, \forall x \in \mathcal{X}$ regression appears to be an appropriate method for estimating the weights.

Empirically, we found that Kernel Ridge Regression (KRR) with an RBF kernel generally yields the smallest mean squared error and the highest stability. This observation suggests that the underlying conditional expectation $w^{*0}$,  possesses the ideal smoothness properties. If $w^{*0}$ were highly irregular, it would contradict our motivation for using it to smooth the weights. Moreover, if the true weight function  $w^0$ or $w^{*0}$ lies in the RKHS associated with the Gaussian kernel, it exhibits an exponential decay in its eigenvalues and can be considered infinitely smooth. Under these conditions, KRR with an RBF kernel can achieve the learning rate of $\operatorname{o}_\P(n^{-1/2})$ \citep{caponnetto2007optimal,smale2007learning}. 
 
Throughout the main text, we use KRR for weight estimation due to its robust empirical performance. A discussion of the performance of alternative methods is provided in \Cref{sec:weights_estimation_models_overview}. Notably, since the regression in this step is univariate, it is easier to get a fast convergence rate than in the higher-dimensional settings.

\section{Experiments}
We include experiments on synthetic data and real-world applications. In each experiment, we run $K$ simulations, and report the results of each estimator $\hat{\theta}$ compared to the true value of ATE $\theta$ regarding bias ($\frac{1}{K}\sum_{k=1}^K\big(\hat{\theta}-\theta\big)$), mean absolute error (MAE, $\frac{1}{K}\sum_{k=1}^K |\hat{\theta}-\theta|$) and root mean squared error (RMSE, $\sqrt{\frac{1}{K}\sum_{k=1}^K\big(\hat{\theta}-\theta\big)^2}$). The code is available at \href{https://github.com/linyingyang/AMR}{https://github.com/linyingyang/AMR}.

\subsection{Synthetic experiment results}
\label{subsec:synthetic_experiment}

In this section, we report empirical results on synthetic data and compare the performance of AMR and MR versus TMLE (Targeted Maximum Likelihood Estimation; \citealp{van2011targeted}), CTMLE (Collaborative Targeted Maximum Likelihood Estimation; \citealp{van2010collaborative}), CBPS (Covariate Balancing Propensity Score; \citealp{imai2014covariate}), AIPW, IPW and DOPE (with the single-index form (\ref{eqn:single-index})).

We employ the \texttt{causl} package built on the frugal parameterization framework \citep{evans2024parameterizing} to generate synthetic data that precisely follows the intended causal distribution. In our simulation, each covariate in $X$ is independently drawn from the normal distribution $\mathcal{N}(0,1)$. We fix the dimensions of the confounders ($C$), prognostic variables ($O$), and spurious variables ($S$) at 5 each, and later assess estimator performance by varying the dimensionality $p_I$ of the instrumental variables ($I$). Their roles can be found in the causal DAG shown in \Cref{fig:w_confounder}. The data generation process is as follows:
\begin{itemize}
    \item $A \mid  X\sim \operatorname{Bernoulli}\left(\pi\left(X\right)\right)$, where $\pi(X) = \operatorname{expit}\left(\sum_{j=1}^{p_{I}}I_j+0.5\sum_{j=1}^{5}C_j\right)$;
    \item $Y \mid A,X\sim\mathcal{N}\left(5A+ \mu^0\left(X\right),1\right)$;
    \item $\mu^0(X) = 10 \operatorname{sin}\left(\pi\sum_{j=1}^{5}O_j\right) + 20\left(\sum_{j=1}^{5}O_j\right)^2 + \left(\sum_{j=1}^{5}C_j\operatorname{cos}\left(\pi\sum_{j=1}^{5}O_j\right)\right)^2$.
\end{itemize}

We design the potential outcome, $\mu^0(X)$, to follow a complex nonlinear form. As shown later, methods that rely on a priori assumptions like DOPE can be substantially biased when these assumptions are violated. Moreover, our design implies that increasing $p_I$ makes the propensity score, $\pi(X)$, increasingly extreme.

The propensity score $\pi(X)$ is estimated using logistic regression. For a fair comparison with DOPE, the potential outcome model is estimated via a 3-layer feedforward neural network. Furthermore, the initial estimates for the propensity score and potential outcome, if needed, are set identical in all estimators. We perform estimation over $K=200$ simulations, varying the sample size $n\in\{400,600,1000\}$ and the instrumental variable dimension $p_I\in\{1,5,10,15,20\}$.

\Cref{fig:synthetic_bias} presents the bias, MAE, and RMSE of these estimators across various settings, with detailed values provided in \Crefrange{tab:synthetic_bias}{tab:synthetic_rmse}. Notably, AMR demonstrates superior performance in terms of MAE and RMSE across all settings, while AIPW is a lot more variable. When $n$ is small, AMR exhibits larger bias due to the regularization in weights estimation; however, the bias becomes comparable as $n$ increases. \Cref{fig:synthetic_bias_new} provides comparison with IPW and AIPW dropped, as they perform poorly compared to the other methods. DOPE is also excluded from \Cref{fig:synthetic_bias} and \Cref{fig:synthetic_bias_new} as it is markedly inferior in this setting, confirming that the misspecification---stemming from $\mu^0(X)$ not adhering to the single-index form in (\ref{eqn:single-index})---induces significant estimation bias.

\begin{figure}
    \centering
    \includegraphics[width=1\linewidth]{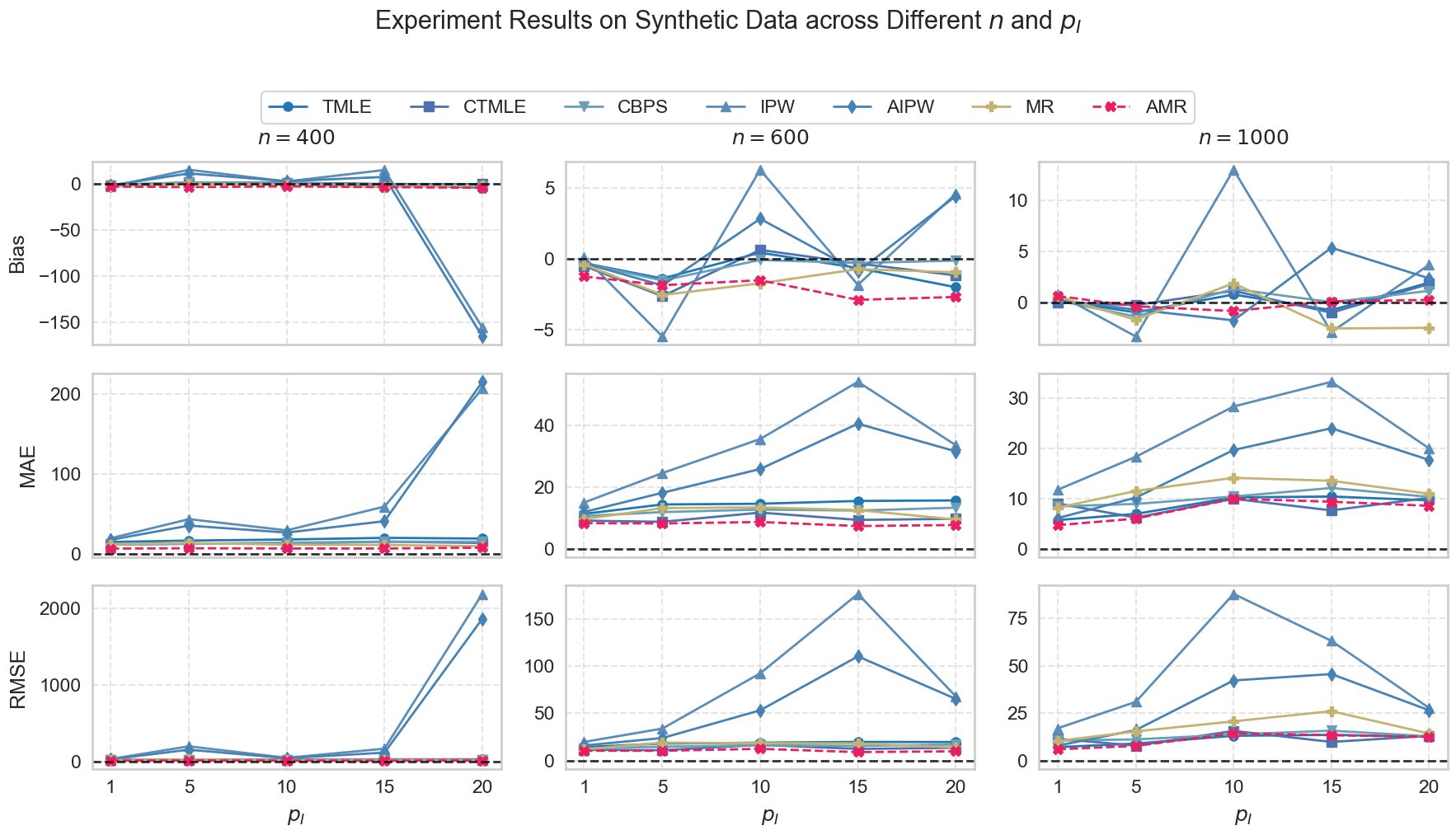}
    \caption{Bias, MAE and RMSE results on synthetic data, varying across $p_I$. }
    \label{fig:synthetic_bias}
\end{figure}
\begin{figure}
    \centering
    \includegraphics[width=1\linewidth]{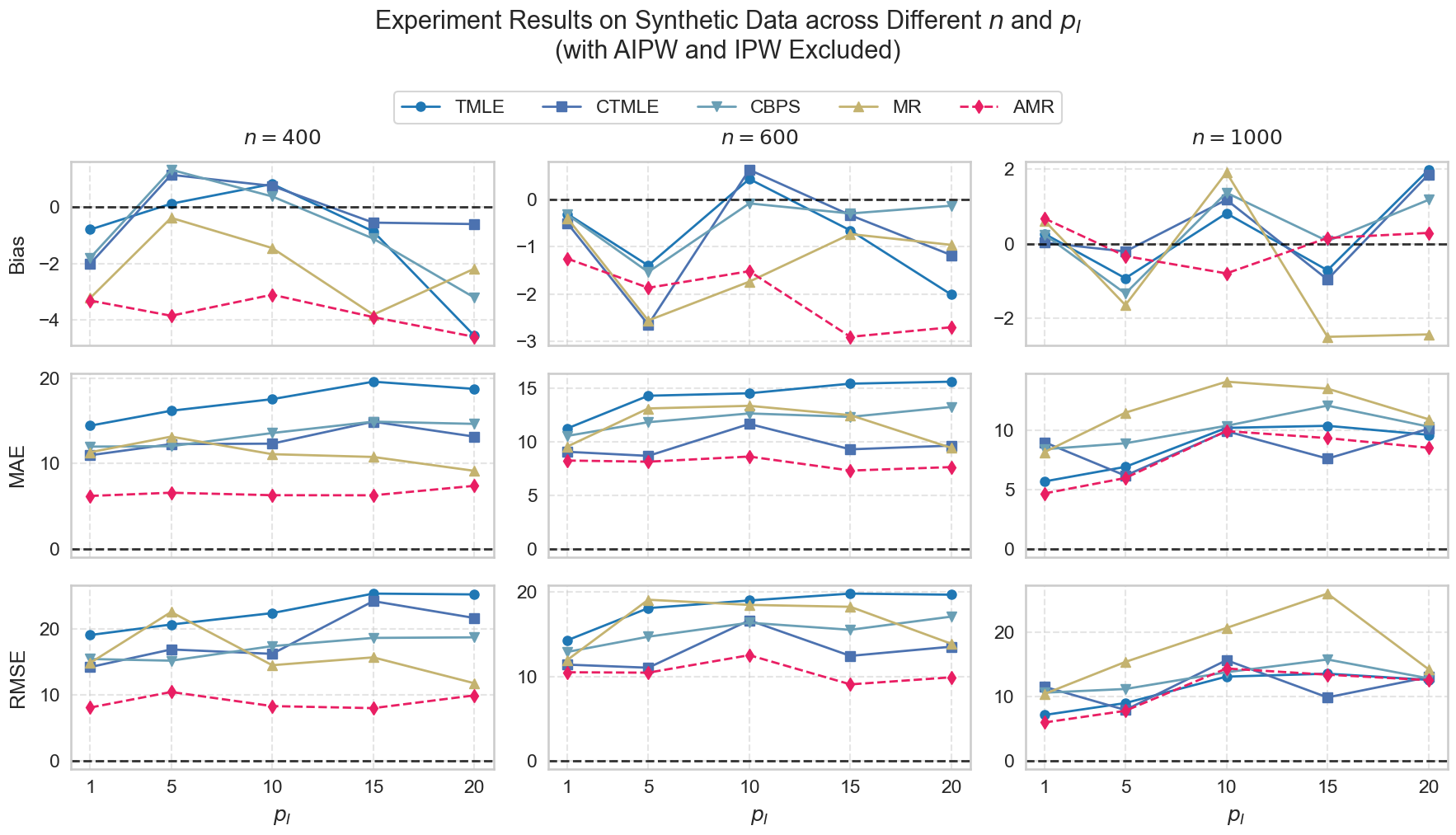}
    \caption{Bias, MAE and RMSE results on synthetic data, varying across $p_I$, with IPW and AIPW excluded. }
    \label{fig:synthetic_bias_new}
\end{figure}
\begin{table}[ht]
\begin{center}
\begin{tabular}{rrrrrrrrrr}
\toprule
 &  & TMLE & CTMLE & CBPS & IPW & AIPW & DOPE & MR & AMR \\
$n$ & $p_I $&  &  &  &  &  &  &  &  \\
\midrule
\multirow[t]{5}{*}{400} & 1 & $ -0.795 $ & $ -2.01 $ & $ -1.80 $ & $ -3.25 $ & $ -1.99 $ & $ 103 $ & $ -3.23 $ & $ -3.32 $ \\
 & 5 & $ 0.121 $ & $ 1.15 $ & $ 1.33 $ & $ 14.9 $ & $ 10.9 $ & $ 103 $ & $ -0.380 $ & $ -3.86 $ \\
 & 10 & $ 0.823 $ & $ 0.754 $ & $ 0.381 $ & $ 2.23 $ & $ 2.52 $ & $ 103 $ & $ -1.45 $ & $ -3.11 $ \\
 & 15 & $ -0.862 $ & $ -0.55 $ & $ -1.12 $ & $ 14.6 $ & $ 7.14 $ & $ 103 $ & $ -3.82 $ & $ -3.91 $ \\
 & 20 & $ -4.56 $ & $ -0.601 $ & $ -3.21 $ & $ -156 $ & $ -165 $ & $ 104 $ & $ -2.19 $ & $ -4.61 $ \\
\cline{1-10}
\multirow[t]{5}{*}{600} & 1 & $ -0.305 $ & $ -0.501 $ & $ -0.322 $ & $ 0.118 $ & $ -0.309 $ & $ 103 $ & $ -0.396 $ & $ -1.25 $ \\
 & 5 & $ -1.40 $ & $ -2.65 $ & $ -1.53 $ & $ -5.50 $ & $ -1.88 $ & $ 102 $ & $ -2.56 $ & $ -1.87 $ \\
 & 10 & $ 0.435 $ & $ 0.630 $ & $ -0.0821 $ & $ 6.30 $ & $ 2.83 $ & $ 104 $ & $ -1.74 $ & $ -1.52 $ \\
 & 15 & $ -0.654 $ & $ -0.331 $ & $ -0.295 $ & $ -1.84 $ & $ -0.904 $ & $ 103 $ & $ -0.732 $ & $ -2.91 $ \\
 & 20 & $ -2.01 $ & $ -1.18 $ & $ -0.131 $ & $ 4.60 $ & $ 4.44 $ & $ 102 $ & $ -0.960 $ & $ -2.70 $ \\
\cline{1-10}
\multirow[t]{5}{*}{1000} & 1 & $ 0.263 $ & $ 0.0283 $ & $ 0.242 $ & $ 0.878 $ & $ 0.282 $ & $ 102 $ & $ 0.614 $ & $ 0.675 $ \\
 & 5 & $ -0.924 $ & $ -0.209 $ & $ -1.33 $ & $ -3.27 $ & $ -0.653 $ & $ 103 $ & $ -1.64 $ & $ -0.331 $ \\
 & 10 & $ 0.813 $ & $ 1.17 $ & $ 1.36 $ & $ 13.0 $ & $ -1.70 $ & $ 103 $ & $ 1.90 $ & $ -0.795 $ \\
 & 15 & $ -0.711 $ & $ -0.958 $ & $ 0.0696 $ & $ -2.84 $ & $ 5.37 $ & $ 103 $ & $ -2.49 $ & $ 0.153 $ \\
 & 20 & $ 1.98 $ & $ 1.85 $ & $ 1.17 $ & $ 3.72 $ & $ 2.39 $ & $ 103 $ & $ -2.43 $ & $ 0.289 $ \\
\bottomrule
\caption{Bias on synthetic data.}
\label{tab:synthetic_bias}
\end{tabular}

\end{center}
\end{table}
\begin{table}[ht]
\begin{center}
\begin{tabular}{rrrrrrrrrr}
\toprule
 &  & TMLE & CTMLE & CBPS & IPW & AIPW & DOPE & MR & AMR \\
$n$ & $p_I$ &  &  &  &  &  &  &  &  \\
\midrule
\multirow[t]{5}{*}{400} & 1 & $ 14.5 $ & $ 11.0 $ & $ 12.0 $ & $ 19.2 $ & $ 17.1 $ & $ 103 $ & $ 11.3 $ & $ 6.20 $ \\
 & 5 & $ 16.2 $ & $ 12.3 $ & $ 12.0 $ & $ 43.1 $ & $ 35.0 $ & $ 103 $ & $ 13.2 $ & $ 6.58 $ \\
 & 10 & $ 17.6 $ & $ 12.3 $ & $ 13.6 $ & $ 29.1 $ & $ 26.3 $ & $ 103 $ & $ 11.1 $ & $ 6.29 $ \\
 & 15 & $ 19.6 $ & $ 14.9 $ & $ 14.9 $ & $ 58.8 $ & $ 40.5 $ & $ 103 $ & $ 10.8 $ & $ 6.28 $ \\
 & 20 & $ 18.8 $ & $ 13.2 $ & $ 14.7 $ & $ 207 $ & $ 215 $ & $ 104 $ & $ 9.15 $ & $ 7.39 $ \\
\cline{1-10}
\multirow[t]{5}{*}{600} & 1 & $ 11.3 $ & $ 9.06 $ & $ 10.6 $ & $ 14.9 $ & $ 11.9 $ & $ 103 $ & $ 9.56 $ & $ 8.25 $ \\
 & 5 & $ 14.3 $ & $ 8.68 $ & $ 11.8 $ & $ 24.5 $ & $ 18.1 $ & $ 102 $ & $ 13.1 $ & $ 8.14 $ \\
 & 10 & $ 14.5 $ & $ 11.7 $ & $ 12.6 $ & $ 35.5 $ & $ 25.8 $ & $ 104 $ & $ 13.4 $ & $ 8.62 $ \\
 & 15 & $ 15.4 $ & $ 9.30 $ & $ 12.3 $ & $ 54.1 $ & $ 40.5 $ & $ 103 $ & $ 12.5 $ & $ 7.31 $ \\
 & 20 & $ 15.6 $ & $ 9.64 $ & $ 13.3 $ & $ 33.5 $ & $ 31.5 $ & $ 102 $ & $ 9.43 $ & $ 7.63 $ \\
\cline{1-10}
\multirow[t]{5}{*}{1000} & 1 & $ 5.70 $ & $ 9.01 $ & $ 8.42 $ & $ 11.7 $ & $ 6.08 $ & $ 102 $ & $ 8.17 $ & $ 4.68 $ \\
 & 5 & $ 6.93 $ & $ 6.19 $ & $ 8.93 $ & $ 18.3 $ & $ 10.2 $ & $ 103 $ & $ 11.5 $ & $ 5.99 $ \\
 & 10 & $ 10.2 $ & $ 9.95 $ & $ 10.4 $ & $ 28.4 $ & $ 19.7 $ & $ 103 $ & $ 14.1 $ & $ 9.96 $ \\
 & 15 & $ 10.4 $ & $ 7.63 $ & $ 12.1 $ & $ 33.3 $ & $ 24.1 $ & $ 103 $ & $ 13.5 $ & $ 9.36 $ \\
 & 20 & $ 9.64 $ & $ 10.1 $ & $ 10.3 $ & $ 20.0 $ & $ 17.7 $ & $ 103 $ & $ 11.0 $ & $ 8.54 $ \\
\bottomrule
\end{tabular}
\caption{MAE on synthetic data.}
\label{tab:synthetic_mae}
\end{center}
\end{table}

\begin{table}[ht]
\begin{center}
\begin{tabular}{rrrrrrrrrr}
\toprule
 &  & TMLE & CTMLE & CBPS & IPW & AIPW & DOPE & MR & AMR \\
$n$ & $p_I$ &  &  &  &  &  &  &  &  \\
\midrule
\multirow[t]{5}{*}{400} & 1 & $ 19.1 $ & $ 14.2 $ & $ 15.4 $ & $ 29.8 $ & $ 26.0 $ & $ 103 $ & $ 14.8 $ & $ 8.06 $ \\
 & 5 & $ 20.7 $ & $ 16.9 $ & $ 15.2 $ & $ 195 $ & $ 151 $ & $ 103 $ & $ 22.6 $ & $ 10.4 $ \\
 & 10 & $ 22.4 $ & $ 16.2 $ & $ 17.4 $ & $ 44.9 $ & $ 40.3 $ & $ 103 $ & $ 14.5 $ & $ 8.28 $ \\
 & 15 & $ 25.4 $ & $ 24.2 $ & $ 18.6 $ & $ 164 $ & $ 113 $ & $ 104 $ & $ 15.7 $ & $ 7.97 $ \\
 & 20 & $ 25.2 $ & $ 21.7 $ & $ 18.7 $ & $ 2190 $ & $ 1860 $ & $ 104 $ & $ 11.8 $ & $ 9.91 $ \\
\cline{1-10}
\multirow[t]{5}{*}{600} & 1 & $ 14.3 $ & $ 11.4 $ & $ 12.9 $ & $ 19.8 $ & $ 16.0 $ & $ 103 $ & $ 12.0 $ & $ 10.5 $ \\
 & 5 & $ 18.1 $ & $ 11.0 $ & $ 14.7 $ & $ 33.9 $ & $ 23.9 $ & $ 103 $ & $ 19.1 $ & $ 10.4 $ \\
 & 10 & $ 19.0 $ & $ 16.7 $ & $ 16.4 $ & $ 92.3 $ & $ 53.2 $ & $ 104 $ & $ 18.5 $ & $ 12.5 $ \\
 & 15 & $ 19.8 $ & $ 12.4 $ & $ 15.5 $ & $ 176 $ & $ 110 $ & $ 103 $ & $ 18.3 $ & $ 9.05 $ \\
 & 20 & $ 19.7 $ & $ 13.5 $ & $ 17.1 $ & $ 68.1 $ & $ 65.3 $ & $ 103 $ & $ 13.9 $ & $ 9.89 $ \\
\cline{1-10}
\multirow[t]{5}{*}{1000} & 1 & $ 7.12 $ & $ 11.5 $ & $ 10.6 $ & $ 17.0 $ & $ 7.90 $ & $ 102 $ & $ 10.4 $ & $ 5.95 $ \\
 & 5 & $ 8.98 $ & $ 7.90 $ & $ 11.2 $ & $ 31.0 $ & $ 16.1 $ & $ 103 $ & $ 15.4 $ & $ 7.78 $ \\
 & 10 & $ 13.1 $ & $ 15.7 $ & $ 13.8 $ & $ 87.8 $ & $ 42.2 $ & $ 103 $ & $ 20.7 $ & $ 14.4 $ \\
 & 15 & $ 13.6 $ & $ 9.87 $ & $ 15.8 $ & $ 63.1 $ & $ 45.6 $ & $ 103 $ & $ 26.1 $ & $ 13.4 $ \\
 & 20 & $ 12.6 $ & $ 13.0 $ & $ 12.8 $ & $ 27.7 $ & $ 26.5 $ & $ 103 $ & $ 14.3 $ & $ 12.6 $ \\
\bottomrule
\end{tabular}
\caption{RMSE on synthetic data.}
\label{tab:synthetic_rmse}
\end{center}
\end{table}

\Cref{fig:synthetic_CI} presents the coverage and lengths of the 95\% confidence intervals. The interval for AMR is conservative and comparable in length to that for AIPW, which aligns with our discussion in \Cref{sec:asymptotic}. In this setting, TMLE and CTMLE yield intervals with under-coverage. 

CBPS also shows competitive performance regarding precision as it relies on estimated propensity score which is correctly specified in this setting. However, the confidence intervals of CBPS are much wider. This can be partially due to the fact that at the balancing stage, CBPS balances all covariates, not just the confounders, which brings more uncertainty. Given the poor performance of DOPE, it is excluded from the plot.

As mentioned in \Cref{sec:asymptotic}, although the confidence interval of AMR constructed using (\ref{eqn:ci_conservative}) is conservative and of similar width to that of $\hat{\theta}_{AIPW}$, AMR is superior to AIPW because it is much less biased in finite samples.  If we reduce the nominal coverage (thus shortening the interval), the biased AIPW ends up with under-coverage, as present in \Cref{fig:synthetic_CI_60} where we construct confidence intervals for 60\% nominal coverage.
\begin{figure}
    \centering
    \includegraphics[width=1\linewidth]{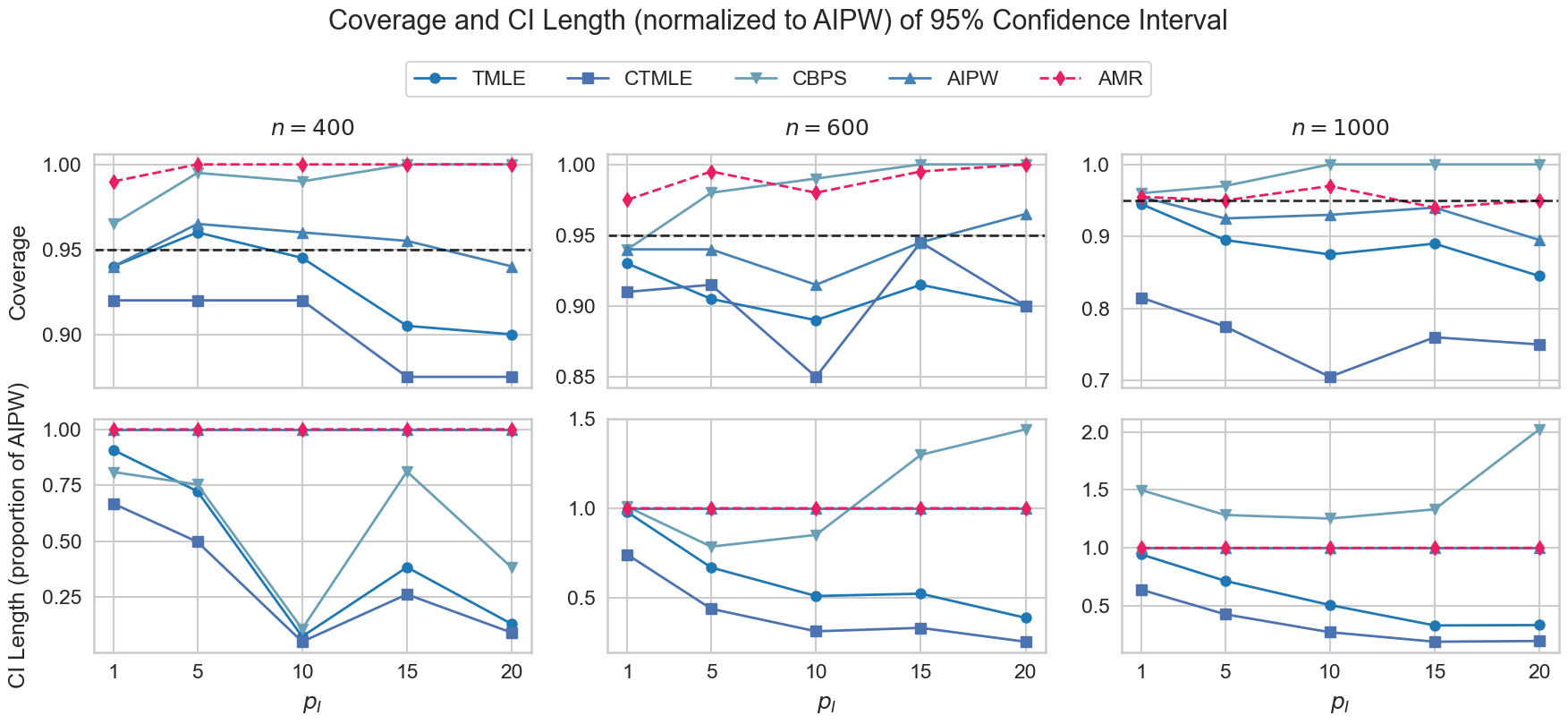}
    \caption{Coverage and length of confidence interval at nominal level of 95\% on synthetic data.}
    \label{fig:synthetic_CI}
\end{figure}

\begin{figure}
    \centering
    \includegraphics[width=1\linewidth]{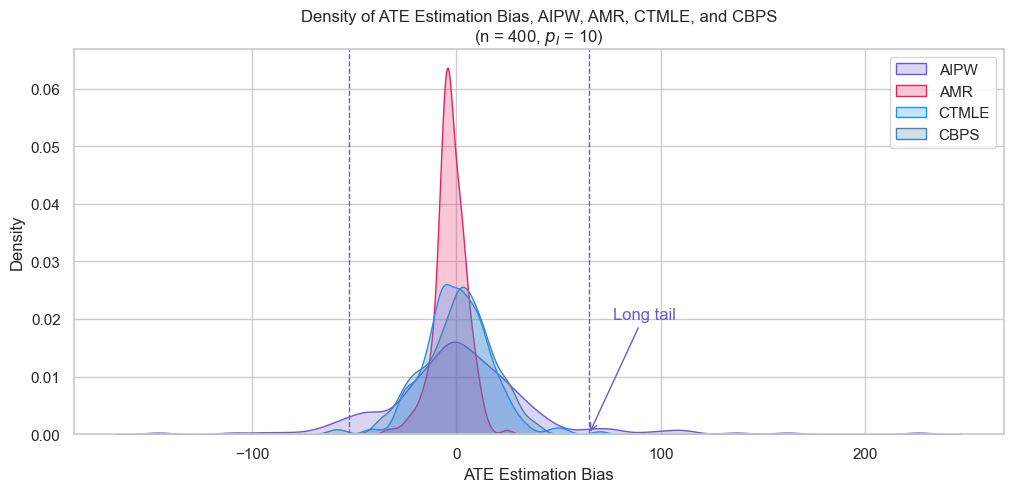}
    \caption{Density of bias, AIPW, AMR, CTMLE and CBPS estimators, across $K=200$ simulations, synthetic data, $n=400$, $p_I=10$.}
    \label{fig:synthetic_density_compare}
\end{figure}

\begin{figure}
    \centering
    \includegraphics[width=1\linewidth]{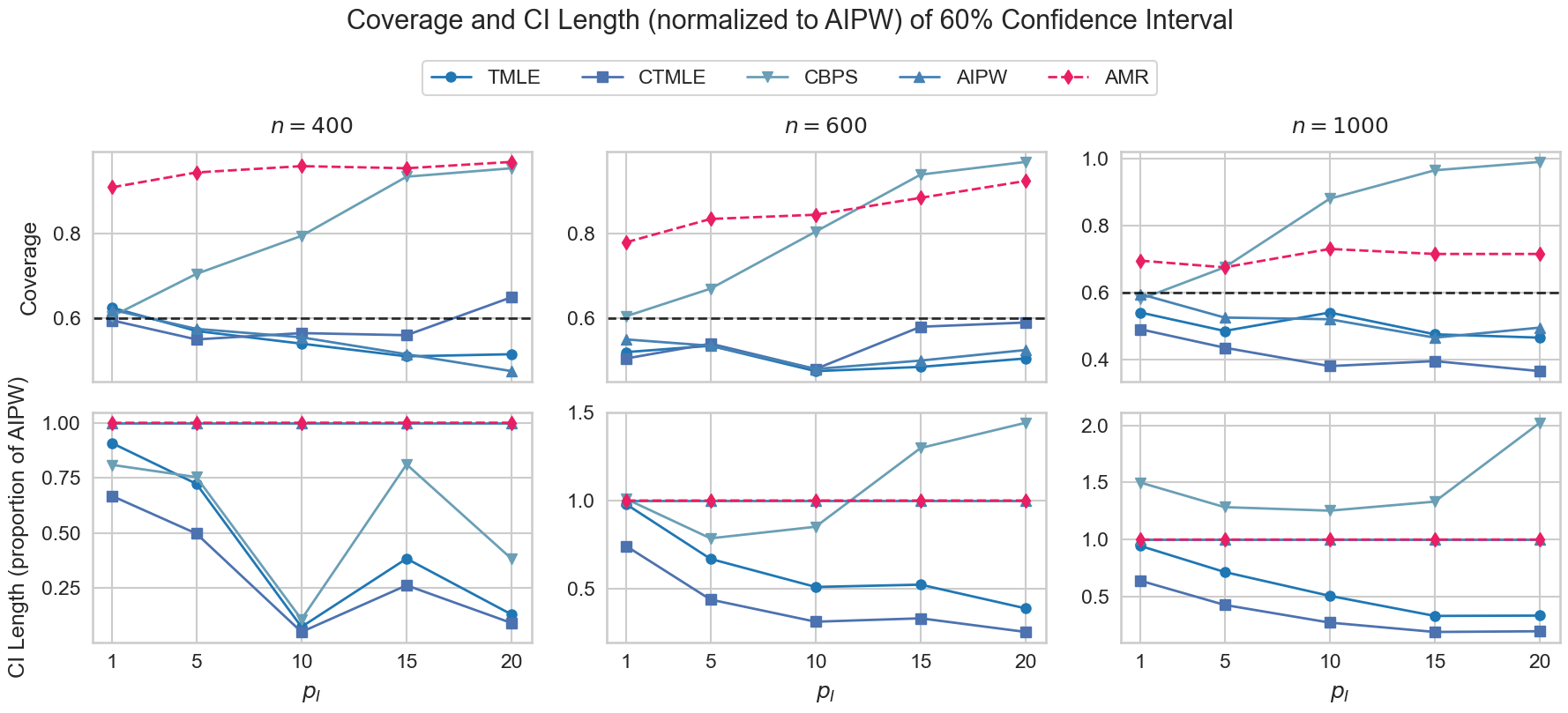}
    \caption{Coverage and length of confidence interval at nominal level of 60\% on synthetic data.}
    \label{fig:synthetic_CI_60}
\end{figure}

\subsection{Application on NHANES data}

\label{subsec:experiment_real}
The National Health and Nutrition Examination Survey (NHANES) is an American nationwide program designed to gather extensive health and nutritional information from a representative sample of the U.S.~population. In our experiment, we use the mortality dataset originally collected in \citet{lundberg2020local}.  In our study, we focus on participants aged 40 to 65 (middle adulthood) and examine the impact of obesity---defined as having a body mass index (BMI) greater than 30---on total cholesterol levels. By dichotomizing BMI at 30, we create a binary intervention indicator ($A=\mathbbm{1}\{\text{BMI}\geq 30\}$), while total cholesterol serves as a continuous outcome measure $Y$, reflecting a key biomarker of cardiovascular health and metabolic risk. Previous research has consistently linked higher BMI with adverse lipid profiles, including elevated total cholesterol \citep{freedman2001relationship,wormser2011emerging, flegal2016trends}, reinforcing the causal plausibility of this exposure-outcome pairing for investigating obesity-related interventions.

We follow the preprocessing steps in \citet{lundberg2020local} and \citet{christgau2024efficient}, dropping covariates with more than 50\% missing values. After selecting the group of people aging between 40 to 65, $n=5985$ observations in total remain for training. To address potential confounding in this observational analysis, we adjust for an extensive set of covariates ($p=61$ in total) that capture a broad spectrum of health indicators. These include demographic and lifestyle factors such as gender, age and physical activity, and laboratory biomarkers including serum albumin, hemoglobin, etc. 

The propensity score is estimated using logistic regression. We present the density of the propensity score estimates in  \Cref{fig:estimated_propen_nhanes}, which indicates lack of overlap as we see a non-trivial proportion of individuals who have a probability of being treated under $0.1$. For comparison, we estimate the potential outcomes using two different models: a 3-layer feed forward neural network (denoted NN) and linear regression (denoted LR) to see how sensitive each method is to the specific potential outcome model. The same set of estimators as in  \Cref{subsec:synthetic_experiment} is compared here. We provide estimates and confidence intervals in \Cref{tab:nhanes_ci}. 

\begin{figure}
    \centering
    \includegraphics[width=0.8\linewidth]{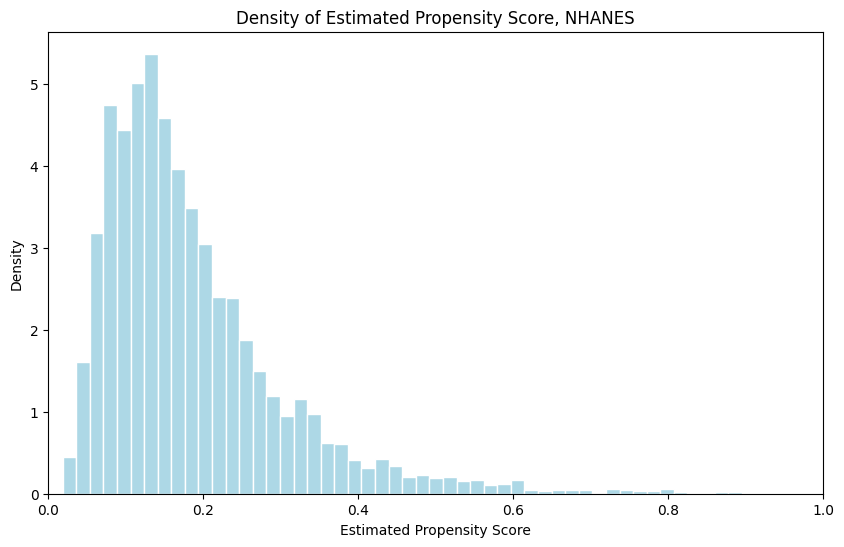}
    \caption{Density of estimated propensity score of exposure $A=\mathbbm{1}\{\text{BMI}\geq 30\}$ on NHANES data.}
    \label{fig:estimated_propen_nhanes}
\end{figure}

\begin{table}[ht]
\centering
\begin{tabular}{lrr}
\toprule
Method & Estimation & 95\% CI \\
\midrule
CBPS   & 4.70    & $[0.426, 8.98]$ \\
DOPE  & 230  & $[-201, 662]$  \\
CTMLE (NN)  & 2.59    & $[-0.511, 5.68]$ \\
CTMLE (LR) & 3.90 & $[0.815, 6.98]$ \\
TMLE (NN)   & 3.41    & $[-0.0175, 6.84]$ \\
TMLE (LR)   &  3.15   & $[-0.315, 6.62]$ \\
AIPW (NN)   & 3.42    & $[-0.00461, 6.85]$ \\
AIPW (LR) &  3.14 & $[-0.326, 6.62]$  \\
AMR (NN)    & 3.44    & $[0.0114, 6.86]$ \\
AMR (LR)    & 3.56    & $[0.0906, 7.03]$ \\
\bottomrule
\end{tabular}
\caption{ATE estimates and 95\% confidence intervals, with 3-layer feedforward neural network as the potential outcome model. The IPW and MR estimates are 3.90, 4.70 respectively.}
\label{tab:nhanes_ci}
\end{table}

Although all methods produce positive estimation as expected, only CBPS, CTMLE with potential outcome estimated by linear regression, and AMR under both potential outcome models have the 95\% confidence intervals bounded above 0. Meanwhile, AMR has the smallest difference in the estimates produced using the linear regression and neural network approaches; this aligns with the fact that we condition on $Y^*$, which is a combination of the observed outcome $Y$ and the bias of $Y - \hat{\mu}^a(X)$ as discussed in \Cref{sec:method}, and hence this is less sensitive to the exact method used to estimate the outcome regressions.

\subsection{Application on text data: News}
\label{subsec:experiment_text}
So far we have demonstrated the good performance of MR and AMR in treatment effect estimation, especially in high-dimensional settings. Another area that will benefit greatly from these properties is treatment effect estimation on text data, where the words themselves are used as covariates. Various  applications of this kind have been explored in recent studies, including \citet{yao2019estimation}, \citet{gui2022causal}, \citet{lin2024isolated}, and \citet{imai2024causal}. Some of these papers specifically discuss mitigation of poor overlap in such datasets. For example, \citet{yao2019estimation} propose a representation learning and matching framework designed to filter out information from nearly-instrumental covariates; \citet{gui2022causal}, as mentioned in \Cref{subsec:related_work}, contribute to this area by proposing a DragonNet-based transformer for potential outcome estimation, and estimating propensity score conditional on estimated potential outcomes. A summary of related papers and datasets is available in \citet{keith2020text}.

We present experimental results on the semi-synthetic News dataset introduced by \citet{johansson2016learning}, which is commonly used in recent causal inference machine learning papers. The News dataset consists of 5000 documents extracted from the NYT Corpus, where each unit represents a news article. The article is consumed either on a model device or on desktop. In the dataset, the opinions of a media consumer exposed to multiple news items are simulated. Each article is characterized by 3477 covariates corresponding to word counts, and the outcome reflects the reader's experience of the article. The intervention $A$ represents the viewing device: whether the article is viewed on a desktop ($A=0$) or a mobile device ($A=1$). The outcome $Y$ represents the reader's experience on an article $X$. The treatments and outcomes were generated using the data generating process described in \citet{johansson2016learning}.

The covariates in this area are high-dimensional and highly sparse, leading to a lack of overlap, which we show with estimated propensity score estimation in \Cref{fig:estimated_propen_news}. In our experiments,  $\hat{\mu}^a(x)$ and $\hat{\pi}(x)$ are estimated using 3-layer feedforward neural network models which might be more suitable in the complex covariates setting. We generate 50 samples from the \textit{News} data-generating process and applied the estimators. Since TMLE, CTMLE and CBPS are extremely time-consuming in this complex context, we do not include them in the performance comparison. 

From the corresponding results shown in \Cref{fig:bias_news} and \Cref{tab:performance_news}, MR and AMR demonstrate significantly better performance than that of IPW and AIPW, with much lower (absolute) bias, MAE and RMSE. In this experiment, DOPE demonstrates relatively good performance compared to IPW and AIPW, but is still worse than MR and AMR, presumably due to the bias in learning the representation of $V_{\hat{\beta}}$. 
\begin{figure}[ht]
    \centering
    \begin{subfigure}[b]{0.45\textwidth}
        \centering
\includegraphics[width=0.82\linewidth]{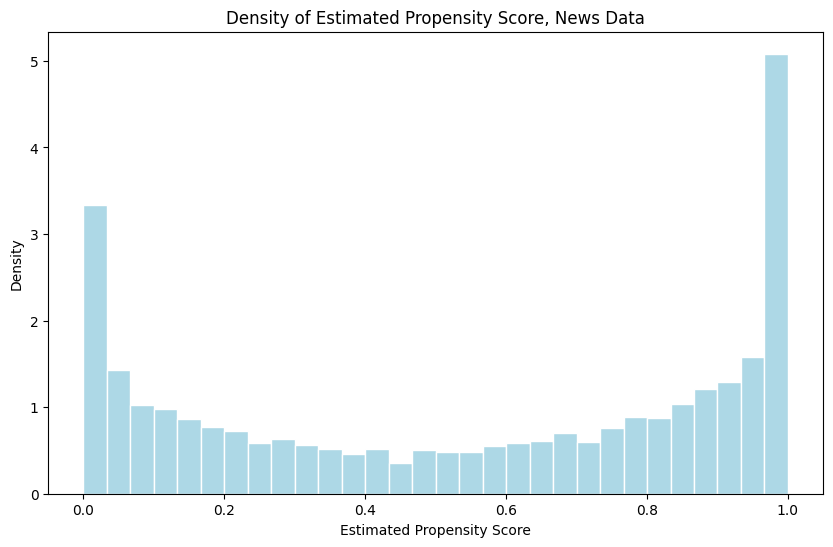}
    \caption{Density of estimated propensity score in \textit{News} dataset, indicating lack of overlap.}
    \label{fig:estimated_propen_news}
    \end{subfigure}
    \hfill
    \begin{subfigure}[b]{0.45\textwidth}
        \centering
\includegraphics[width=.9\linewidth]{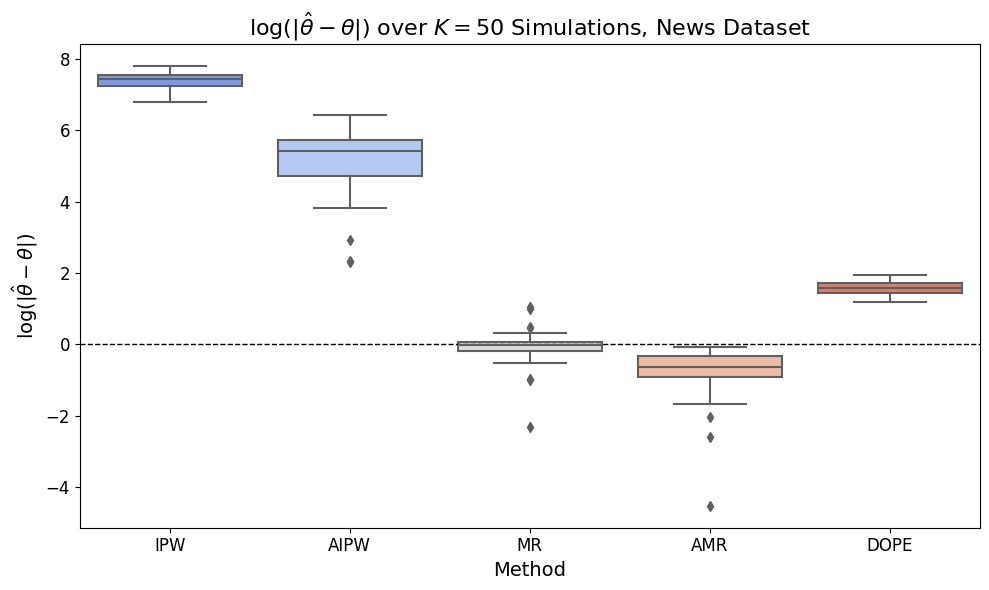}
    \caption{Estimator performance on $K=50$ draws of data generating process in the \textit{News} dataset.}
    \label{fig:bias_news}
    \end{subfigure}
    \caption{Experiment results on \textit{News} data.}
    \label{fig:news_results}
\end{figure}

\begin{table}
\begin{center}
\begin{tabular}{rrrrrr}
\toprule
     & \textbf{IPW} & \textbf{AIPW} & \textbf{MR} & \textbf{AMR} &  \textbf{DOPE}\\
\midrule
Bias & $1660$ & $-225$ & $-0.953$ & $-\textbf{0.532}$ & $4.74$ \\
MAE & $1660$ & $232$ & $1.01$ & $\textbf{0.535}$ & $4.74$ \\
RMSE & $1710$ & $273$ & $1.11$ & $\textbf{0.582}$ & $4.80$ \\
\bottomrule
\end{tabular}
\end{center}
\caption{Estimator performance on $K=50$ draws of data generating process in the \textit{News} dataset.}
\label{tab:performance_news}
\end{table}

\

\
Our estimators are highly flexible and computationally efficient for text applications with various featurization methods. In \Cref{sec:amazon}, we present experiments where text is transformed into embeddings via Transformers \citep{wolf2019huggingface}. We compare these results with the TI estimator \citep{gui2022causal}, a Transformer-based ATE estimator specifically designed for text applications. The findings further demonstrate the robustness of our proposed MR and AMR estimators, particularly highlighting their post-hoc calibration capabilities.

\section{Discussion}\label{sec:discussion}
In this paper, we propose two estimators that exhibit robust performance in the presence of poor overlap, either in finite samples or because of complex, high-dimensional covariates. Our approach leverages outcome-informative weights, integrating outcome data without imposing extra assumptions beyond the estimates of potential outcome means and propensity scores. The post-hoc calibration strategy not only reduces variance but also minimizes risks associated with pre-processing steps, such as covariate selection or representation learning based on predetermined structural assumptions. While our method achieves asymptotic efficiency comparable to conventional approaches like AIPW and IPW, the AMR and MR estimators offer significant advantages in finite samples.

Moreover, the concept of outcome-informative, post-hoc calibrated weights is highly flexible and extendable. For example, we also develop corresponding estimators for the Average Treatment Effect on the Treated (ATT) and the Average Treatment Effect on the Control (ATC):
\begin{definition}[Extension to ATT and ATC]
    Following a similar idea, the AMR estimators for ATT and ATC are:
\begin{itemize}
    \item $\hat{\theta}_{ATT} = \P_n\Bigl\{\hat{\E}\left[\frac{A}{\hat{\pi}(X)}\middle|Y-\left\{1-\hat{\pi}(X)\right\}\hat{\mu}^1(X)\right]\Bigl[Y-\left\{1-\hat{\pi}(X)\right\}\hat{\mu}^1(X)\Bigr]\Bigr\}$,
    \item $\hat{\theta}_{ATC} = \P_n\Bigl\{\hat{\E}\left[\frac{1-A}{1-\hat{\pi}(X)}\middle|Y-\hat{\pi}(X)\hat{\mu}^0(X)\right]\Bigl[Y-\hat{\pi}(X)\hat{\mu}^0(X)\Bigr]\Bigr\}$.
\end{itemize}
    
\end{definition}
We can also complement the policy evaluation in \citet{taufiq2024marginal} with a doubly robust, outcome-informed weighted estimator:
\begin{definition}[Extension to policy evaluation]
    As originated from the policy evaluation context in \cite{taufiq2024marginal}, AMR can be modified to adapt to policy evaluation by writing
$$
\hat{\E}Y^t = \P_n\left\{\hat{\E}\left[\frac{\pi^t(X)}{\hat{\pi}^{obs}(X)} \,\middle|\, Y^{obs}-\hat{\mu}(X,A)\right]\left(Y^{obs}-\hat{\mu}(X,A)\right)\right\}, 
$$
where the superscript $t$ stands for target domain (the domain in which we want to estimate the outcome of implementing a policy), $X$ represents the states and $A$ denotes the action, $\pi$ stands for the policy.
\end{definition}
In \Cref{subsec:connect_efficient_set}, we show how MR relates to existing efficient adjustment set selection methods. However, establishing a comparable link for AMR is less straightforward because the propensity score model, $\pi(X)$, is included in $Y^{*}$, which is something we condition on. Investigating how AMR might connect to the efficient adjustment set search thus remains an intriguing direction for future work. Meanwhile, the confidence intervals built using \Cref{prop:conservative_ci} are conservative, as shown in \Cref{subsec:synthetic_experiment}. We tried  estimating the variance and building a confidence interval for $\hat{\theta}_{AMR}$ using the standard nonparametric bootstrap, but it does not perform well in our case: the resulting confidence intervals are unusually wide in our experiments. This may be attributed to the significant lack of overlap in our data. Given the significant violation of the positivity assumption, any further reduction in overlap only exacerbates the uncertainty. The resampling process inherent to the bootstrap likely intensifies this issue by altering the degree of overlap in the underlying datasets, worsening the positivity violation. Constructing a shorter confidence interval with the desired coverage would be a very useful extension.

\section*{Acknowledgments}
We would like to thank Muhammad Faaiz Taufiq and Jake Fawkes for helpful discussions.  Linying Yang is supported by the EPSRC Centre for Doctoral Training in Modern Statistics and Statistical Machine Learning (EP/S023151/1) and Novartis.  Robin Evans is supported by GSK and the SMARTbiomed Pioneer Centre.

\bibliography{paper-ref}

\newpage

\appendix
\section{Main proofs}

\subsection{Proof of \Cref{prop:variance_comparison_oracle}}
\label{sec:oracle_var_proof}
\begin{proof}[Proof of \Cref{prop:variance_comparison_oracle}]
    This proof is taken from \citet{taufiq2024marginal}.  Assume we have the true weights $w$ and $w^{*}$, and $\hat{\pi} = \pi$, $\hat{\mu}^a = \mu^a, a\in\{0,1\}$. We have $\E\hat{\theta}^0_{MR} = \E\hat{\theta}^0_{AMR} = \E\hat{\theta}^0_{IPW} = \E\hat{\theta}^0_{AIPW}.$ Thus,
\begin{align*}
    \Var \hat{\theta}^0_{IPW} - \Var \hat{\theta}^0_{MR} &= \left(\E\left(\hat{\theta}^0_{IPW}\right)^2 - \left[\E\hat{\theta}^0_{IPW}\right]^2\right) - \left(\E\left(\hat{\theta}^0_{MR}\right)^2 - \left[\E\hat{\theta}^0_{MR}\right]^2\right)\\
    &= \E\left(\hat{\theta}^0_{IPW}\right)^2  - \E\left(\hat{\theta}^0_{MR}\right)^2 \\
    &= \frac{1}{n}\left(\E\left[h(A,X)^2Y^2\right] - \E\left[w(Y)^2Y^2\right] \right) \\
    &= \frac{1}{n}\left(\E\left\{\E\left[h(A,X)^2 \middle| Y\right]Y^2\right\} - \E\left\{\Bigl(\E\left[h(A,X) \middle| Y\right]\Bigr)^2Y^2\right\} \right) \\
    & = \frac{1}{n} \E\left\{\Var\left[h(A,X) \middle| Y\right]Y^2\right\} \geq 0.
\end{align*}
By replacing $Y$ with ${Y}^{*}$, we obtain $\Var \hat{\theta}^0_{AIPW} - \Var \hat{\theta}^0_{AMR} = \frac{1}{n} \E\left\{\Var\left[h(A,X)\middle|{Y}^{*}\right]{Y^{*}}^2\right\} \geq 0$ similarly.
\end{proof}

\subsection{Proof of \Cref{prop:variance_comparison_proposed}}
\label{sec:var_proof}
\begin{proof}[Proof of \Cref{prop:variance_comparison_proposed}]
    We first investigate $\Var \hat{\theta}_{IPW}-\Var \hat{\theta}_{AIPW}$:
\begin{align*}
    \Var \hat{\theta}_{IPW}-\Var \hat{\theta}_{AIPW} &= \left(\E\hat{\theta}^2_{IPW} - \left[\E\hat{\theta}_{IPW}\right]^2\right)  - \left(\E\hat{\theta}_{AIPW}^2 - \left[\E\hat{\theta}_{AIPW}\right]^2\right) \\
    &= \E\hat{\theta}^2_{IPW} - \E\hat{\theta}^2_{AIPW}\\
    & = \frac{1}{n}\left(\E\left[ h(A,X)^2 Y^2\right] - \E \left[h(A,X)^2 {Y^{*}}^2\right]\right) \\
    & = \frac{1}{n}\E\big[h(A,X)^2 (Y^2-{Y^{*}}^2)\big].
\end{align*}
Then we have
\begin{align*}
    \lefteqn{\Var \hat{\theta}^0_{MR} - \Var \hat{\theta}^0_{AMR}}\\
    &= \left(\Var\hat{\theta}^0_{MR} - \Var \hat{\theta}^0_{IPW}\right) + \left(\Var \hat{\theta}^0_{IPW} - \Var \hat{\theta}^0_{AIPW}\right) + \left(\Var \hat{\theta}^0_{AIPW} - \Var \hat{\theta}^0_{AMR}\right) \\
    &= \frac{1}{n}\E\left\{ - \Var\left[h(A,X)\middle|Y\right]Y^2 + h(A,X)^2 (Y^2-{Y^{*}}^2) + \Var\left[h(A,X)\middle|{Y}^{*}\right]{Y^{*}}^2\right\}\\
    &= \frac{1}{n}\E\left\{\Var\left[h(A,X)\middle|{Y}^{*}\right]{Y^{*}}^2 - \Var\left[h(A,X)\middle|Y\right]Y^2 + h(A,X)^2 (Y^2-{Y^{*}}^2)\right\}\\
    & \geq \frac{1}{n}\E\left\{ \left[h(A,X)^2 - \max\left(\Var\left[h(A,X)\middle|{Y}^{*}\right], \Var\left[h(A,X)|Y\right]\right) \right](Y^2 - {Y^{*}}^2)\right\}.
\end{align*}
\end{proof}

\subsection{Proof of \Cref{thm:DR-AMR}}\label{sec:double_robustness}
\begin{proof}[Proof of Theorem~\ref{thm:DR-AMR}]
\label{proof:DR-AMR}
We would like to prove 
$$
\hat{\theta}_{AMR} \;\overset{P}{\longrightarrow}\; \theta \;
$$
when we have either $\lVert \hat{\pi} - \pi \rVert \;\overset{P}{\longrightarrow}\;0\;$ or $\lVert \hat{\mu}^a - \mu^a \rVert \;\overset{P}{\longrightarrow}\;0,\; a\in\{0,1\}$.

We decompose $\hat{\theta}_{AMR}-\theta$ as
\begin{align*}
    \hat{\theta}_{AMR} - \theta &= \P_n\left[\hat{w}^{*}(\hat{Y}^{*})\hat{Y}^{*}\right] - \theta\\
    &= \underbrace{\P_n\left[\hat{w}^{*}(\hat{Y}^{*})\hat{Y}^{*} - w^{*0}(\hat{Y}^{*})\hat{Y}^{*}\right]}_{\text{(A)}} + \underbrace{\P_n\left[w^{*0}(\hat{Y}^{*})\hat{Y}^{*} - \hat{h}(A,X)\hat{Y}^{*}\right]}_{\text{(B)}}  \\
    & \quad + \underbrace{\P_n\left[\hat{h}(A,X)\hat{Y}^{*}\right]-\theta}_{\text{(C)}}.
\end{align*}

For (A), using the Cauchy-Schwarz inequality, we get
$$
\Big\vert \;\P_n\Big[\hat{w}^{*}(\hat{Y}^{*})\hat{Y}^{*} - w^{*0}(\hat{Y}^{*})\hat{Y}^{*}\Big]\;\Big\vert \;\leq\;\big\lVert\hat{w}^{*}(\hat{Y}^{*}) - w^{*0}(\hat{Y}^{*})\big\rVert_{n,2} \;\big\lVert\hat{Y}^{*}\big\rVert_{n,2}.
$$
where $\lVert f \rVert_{n,2} = (\int f(z)^2d\P_n(z))^{1/2}$ denotes the empirical norm.
With the consistency assumption of $\hat{w}^*$ to $w^{*0}$ and boundedness assumption  of $\hat{Y}^{*}$ in \Cref{thm:DR-AMR}, the RHS tends to zero in probability.
Thus (A) $\overset{P}{\longrightarrow} \;0$ as $n\;\rightarrow\;\infty$.

For (B), note that we have 
\begin{align} 
    \E\big[w^{*0}(\hat{Y}^{*})\hat{Y}^{*}\big] \;&= \;\E\big[\E\big[\hat{h}(A,X) \big| \hat{Y}^{*}\big]\hat{Y}^{*}\big] \nonumber\\
    &=\;\E\big[\E\big[\hat{h}(A,X)\hat{Y}^{*}\big | \hat{Y}^{*}\big]\big] \nonumber\\
    &=\;\E\big[\hat{h}(A,X)\hat{Y}^{*}\big]. \label{eqn:AMR0_exp}
\end{align}
Thus $\E\big[w^{*0}(\hat{Y}^{*})\hat{Y}^{*} - \hat{h}(A,X)\hat{Y}^{*}\big]=0$. 
Using the boundedness assumption in \Cref{thm:DR-AMR}, there exists a constant $M>0$, so that
$$\E\big[w^{*0}(\hat{Y}^{*})\hat{Y}^{*} - \hat{h}(A,X)\hat{Y}^{*}\big]^2 < M,$$
and by the weak law of large numbers, we obtain
$$
\text{(B)} \coloneqq \P_n\big[w^{*0}(\hat{Y}^{*})\hat{Y}^{*} - \hat{h}(A,X)\hat{Y}^{*}\big] \;\overset{P}{\longrightarrow}\; 0
$$
as $n \to \infty$.

(C) can be expressed as $\hat{\theta}_{AIPW} - \theta$, which we know converges to $0$ when either $\hat{\pi}$ or $\hat{\mu}^a$ is consistent, because AIPW is well-known to be doubly robust. 

Combining the results of (A), (B) and (C), we get the consistency of $\hat{\theta}_{AMR}$ to $\theta$, with only one of the propensity score model or outcome regression models being consistent. 
\end{proof}

\subsection{Proof of \Cref{thm:asymp_normal}}
\label{sec:proof_asymp_normal}
\begin{proof}[Proof of \Cref{thm:asymp_normal}]
We provide two separate proofs. 

The first one follows Theorem 5.31 in \citet{van2000asymptotic}:
\begin{equation}
\label{eqn:drift}
    \hat{\theta}^0_{AMR} - \theta =\E\varphi_{AMR}\big(Z;\theta,\hat{\xi}\big) + (\P_n-\E)\varphi_{AMR}(Z;\theta,\xi) + \operatorname{o}_\P\!\big\{n^{-1/2}+\lvert \E\varphi_{AMR}\big(Z;\theta,\hat{\xi}\big) \rvert \big\}.
\end{equation}
As in (\ref{eqn:AMR0_exp}), we have 
\begin{align*}
    \E\varphi_{AMR}\big(Z;\theta,\hat{\xi}\big) =\E\varphi_{AIPW}\big(Z;\theta,\hat{\xi}\big) .
\end{align*}
Thus we can leverage the proof of asymptotic normality of $\hat{\theta}_{AIPW}$. When $\sqrt{n}(\hat{\theta}_{AIPW}-\theta)\;\rightsquigarrow\;\mathcal{N}(0,\Var\{\psi_{AIPW}(Z;\theta,\xi)\})$ \citep{chernozhukov2018double,kennedy2024semiparametric}, we have
$$
\E\varphi_{AMR}\big(Z;\theta,\hat{\xi}\big) = \E\varphi_{AIPW}\big(Z;\theta,\hat{\xi}\big) = \operatorname{o}_\P(n^{-1/2}).
$$
Plugging this back into (\ref{eqn:drift}), we get
$$
\hat{\theta}^0_{AMR} - \theta  = (\P_n-\E)\varphi_{AMR}(Z;\theta,\xi) + \operatorname{o}_\P(n^{-1/2}).
$$
Thus we have $\sqrt{n}\big(\hat{\theta}^0_{AMR} - \theta\big) \;\rightsquigarrow\; \mathcal{N}\big(0, \sigma^2_{AMR,0}\big)$ where $\sigma^2_{AMR,0} = \Var\left\{\varphi_{AMR}(Z;\theta,\xi)\right\}$.

The other approach is to use $\E\hat{\theta}^0_{AMR} = \E\big[w^{*0}(\hat{Y}^{*})\hat{Y}^{*}\big]=\E\big[\hat{h}(A,X)\hat{Y}^{*}\big]=\E\hat{\theta}_{AIPW}$ that is proved in (\ref{eqn:AMR0_exp}) as the bridge to connect $\hat{\theta}^0_{AMR}$ with $\hat{\theta}_{AIPW}$. 
We write 
\begin{align*}
    \hat{\theta}^0_{AMR} - \theta
    &= \hat{\theta}^0_{AMR} - \E\hat{\theta}^0_{AMR}+\E\hat{\theta}^0_{AMR} - \hat{\theta}_{AIPW} + \hat{\theta}_{AIPW}- \theta\\
    &= \underbrace{\hat{\theta}^0_{AMR} - \E\hat{\theta}^0_{AMR}}_{\text{(A)}}+\underbrace{\E\hat{\theta}_{AIPW} - \hat{\theta}_{AIPW}}_{\text{(B)}} + \underbrace{\hat{\theta}_{AIPW}- \theta}_{\text{(C)}}\;
\end{align*}

Using the central limit theorem, we have 
$$
\sqrt{n}\big(\hat{\theta}^0_{AMR} - \E\hat{\theta}^0_{AMR}\big)\;\rightsquigarrow\; \mathcal{N}\big(0,\Var\hat{\theta}^0_{AMR}\big)
$$
and 
$$
\sqrt{n}\big(\hat{\theta}_{AIPW} - \E\hat{\theta}_{AIPW}\big)\;\rightsquigarrow\;\mathcal{N}\big(0,\Var\hat{\theta}_{AIPW}\big).
$$ We also have 
$$ 
\sqrt{n}\big(\hat{\theta}_{AIPW} - \theta\big) \;\rightsquigarrow\;
 \mathcal{N}(0, \Var\{\varphi_{AIPW}(Z;\theta,\xi)\})
 $$ when we have conditions for asymptotic normality of $\hat{\theta}_{AIPW}$ satisfied.

Combining (A), (B) and (C),
we get $\sqrt{n}\big(\hat{\theta}^0_{AMR}-\theta\big)\;\rightsquigarrow\; \mathcal
{N}\big(0,\sigma^2_{AMR,0}\big)$, but we still need to show that $\sigma^2_{AMR,0} = \Var\{\varphi_{AMR}(Z;\theta,\xi)\}$, which is not obvious using this method.
\end{proof}
\subsection{Proof of \Cref{prop:eb_comparison}}
\begin{proof}[Proof of \Cref{prop:eb_comparison}]
The proof is similar to the proof of \Cref{prop:variance_comparison_oracle}, but we give a different derivation here. 

As we know that $\E\left[h(A,X)Y^*\right] = \E\left[w^*(Y^*)Y^*\right]= \theta$, we have
\begin{align*}
    \Var\{\varphi_{AIPW}(Z;\theta,\xi)\} -     \Var\{\varphi_{AMR}(Z;\theta,\xi)\} &=     \E\left[h(A,X)Y^*-\theta\right]^2-\E\left[w^*(Y^*)Y^*-\theta\right]^2 \\
    &=\E\left[h(A,X)Y^*\right]^2-\E\left[w^*(Y^*)Y^*\right]^2\\
    &= \Var\{h(A,X)Y^*\}-\Var\{w^*(Y^*)Y^*\}.
\end{align*}

Remember that $w^*(Y^*)=\E[h(A,X)|Y^*]$ and $\E\left\{\left[h(A,X)-w^*(Y^*)\right]f(Y^*)\right\}=0$ for any $f(Y^*)$ that is a function of $Y^*$ only. We look at the term $\Var\{h(A,X)Y^*\}$:
\begin{align*}
    \Var\{h(A,X)Y^*\}&=\Var\{h(A,X)Y^*-w^*(Y^*)Y^*+w^*(Y^*)Y^*\}\\
    &=\Var\left\{w^*(Y^*)Y^*\right\} + \Var\left\{h(A,X)Y^*-w^*(Y^*)Y^*\right\} \\
    &\qquad\qquad + 2\Cov[h(A,X)Y^*-w^*(Y^*)Y^*, w^*(Y^*)Y^*].
\end{align*}
We know that $\E[h(A,X)Y^*-w^*(Y^*)Y^*]=0$. Thus, 
\begin{align*}
    \Cov[h(A,X)Y^*-w^*(Y^*)Y^*, w^*(Y^*)Y^*] &= \E\big\{\!\left[h(A,X)Y^*-w^*(Y^*)Y^* \right] w^*(Y^*)Y^*\big\}\\
    &= \E\big\{\!\left[h(A,X)-w^*(Y^*) \right] w^*(Y^*){Y^*}^2\big\}\\
    &= \E\big\{\!\left[h(A,X)-w^*(Y^*) \right] f(Y^*)\big\}\\
    &=0.
\end{align*}
Thus
\begin{align*}
    \lefteqn{\Var\{\varphi_{AIPW}(Z;\theta,\xi)\} -     \Var\{\varphi_{AMR}(Z;\theta,\xi)\}}\\ 
    &= 
 \Var\{h(A,X)Y^*\}-\Var\{w^*(Y^*)Y^*\}\\
 &=\Var\left\{w^*(Y^*)Y^*\right\} + \Var\left\{h(A,X)Y^*-w^*(Y^*)Y^*\right\}-\Var\left\{w^*(Y^*)Y^*\right\}\\
&=\Var\left\{h(A,X)Y^*-w^*(Y^*)Y^*\right\}\geq0.
\end{align*}
The difference, $\Var\left\{h(A,X)Y^*-w^*(Y^*)Y^*\right\}$ is zero only when either $\P\big\{h(A,X)=\E\left[h(A,X)\middle| Y^*\right]\!\big\}=1$ or $\P\left(Y^*=0\right)=1$.
\end{proof}

\subsection{Proof of \Cref{thm:asymp_normal_strict}}
\begin{proof}[Proof of \Cref{thm:asymp_normal_strict}]
Based on the proof above, we further denote $\epsilon^{*}(\hat{Y}^{*}) \coloneqq \hat{w}^{*}(\hat{Y}^{*})-w^{*0}(\hat{Y}^{*})$. We make this decomposition:
\begin{align}
\label{eqn:amr_decomp}
    \hat{\theta}_{AMR} - \theta &= \P_n\big[\hat{w}^{*}(\hat{Y}^{*})\hat{Y}^{*}\big] - \theta \nonumber \\
    &= \P_n\big[(w^{*0}(\hat{Y}^{*})+\epsilon^{*}(\hat{Y}^{*}) )\hat{Y}^{*}\big] - \theta \nonumber \\
    &=\underbrace{\P_n\big[\epsilon^{*}(\hat{Y}^{*})\hat{Y}^{*}\big]}_{\text{(A)}} + \underbrace{\P_n\big[w^{*0}(\hat{Y}^{*})\hat{Y}^{*}\big] - \theta}_{\text{(B)}}.
\end{align}
We have assumed $\E \hat{Y}^{*2} < M$ and $\big\lVert\hat{w}^{*}(\hat{Y}^{*})-w^{*0}(\hat{Y}^{*})\big\rVert_2 = \operatorname{o}_\P(n^{-1/2})$. Thus, using the Cauchy-Schwarz inequality, we get
\begin{align*}
  \big\vert \P_n\big[\hat{w}^{*}(\hat{Y}^{*})\hat{Y}^{*} - w^{*0}(\hat{Y}^{*})\hat{Y}^{*}\big]\big\vert \;&\leq\;\big\lVert\hat{w}^{*}(\hat{Y}^{*}) - w^{*0}(\hat{Y}^{*})\big\rVert_2 \;\big\lVert\hat{Y}^{*}\big\rVert_2\;\\
    & = \operatorname{o}_\P(n^{-1/2}). 
\end{align*}
Given we have the asymptotic normality of (B) in \Cref{thm:asymp_normal}, $\sqrt{n}\big(\hat{\theta}_{AMR}-\theta\big) \;\rightsquigarrow\; \mathcal{N}\big(0, \sigma^2_{AMR,0}\big)$ follows.
\end{proof}

\subsection{Proof of \Cref{prop:conservative_ci}}
\begin{proof}[Proof of \Cref{prop:conservative_ci}]
With the consistency of $\hat{\theta}_{AIPW}$, $\hat{\theta}^0_{AMR}$ and $\hat{\theta}_{AMR}$ to $\theta$, we obtain that $(\hat{\sigma}'_{AMR})^2 = \P_n\big[\hat{h}(A,X)\hat{Y}^* - \hat{\theta}_{AMR}\big]^2$ is a consistent estimator of $ \Var\varphi_{AIPW}(Z;\theta,\xi)$ and $\hat{\sigma}^2_{AMR}=\P_n\big[w^*(\hat{Y}^*)\hat{Y}^* - \hat{\theta}_{AMR}\big]^2$ a consistent estimator of $\sigma^2_{AMR,0}=\Var\psi_{AMR}(Z;\theta,\xi)$.
 
Recall that we have
$\sqrt{n}(\hat{\theta}^0_{AMR}-\theta)/\sigma_{AMR,0} \;\rightsquigarrow\; \mathcal{N}\left(0,1\right)$, so using Slutsky's theorem, we have $\sqrt{n}(\hat{\theta}^0_{AMR}-\theta)/\hat{\sigma}_{AMR}\;\rightsquigarrow\; \mathcal{N}\left(0,1\right)$. The Wald-type confidence interval is thus constructed as 
\begin{equation}
        \left[\hat{\theta}^0_{AMR}-z_{1-\alpha/2}\frac{\hat{\sigma}
    _{AMR}}{\sqrt{n}}, \; \hat{\theta}^0_{AMR}+z_{1-\alpha/2}\frac{\hat{\sigma}
    _{AMR}}{\sqrt{n}}\right].
\label{eqn:ci_oracle}
\end{equation}

We do the following transformation of (\ref{eqn:ci_oracle}):
\begin{align*}
    \hat{\theta}_{AMR} - z_{1-\alpha/2}\frac{\hat{\sigma}'_{AMR}}{\sqrt{n}} & = \hat{\theta}_{AMR} - \hat{\theta}^0_{AMR}+\hat{\theta}^0_{AMR} - z_{1-\alpha/2}\frac{\hat{\sigma}'_{AMR}-\hat{\sigma}_{AMR} + \hat{\sigma}_{AMR}}{\sqrt{n}} \\
    & = \underbrace{\hat{\theta}^0_{AMR} - z_{1-\alpha/2}\frac{\hat{\sigma}_{AMR}}{\sqrt{n}}}_{\text{(A$_1$)}} +  \underbrace{\hat{\theta}_{AMR} - \hat{\theta}^0_{AMR}}_{\text{(B)}} -  \underbrace{z_{1-\alpha/2}\frac{\hat{\sigma}'_{AMR}-\hat{\sigma}_{AMR}}{\sqrt{n}}}_{\text{(C)}}\;;\\
    \hat{\theta}_{AMR} + z_{1-\alpha/2}\frac{\hat{\sigma}'_{AMR}}{\sqrt{n}} & = \hat{\theta}_{AMR} - \hat{\theta}^0_{AMR}+\hat{\theta}^0_{AMR} + z_{1-\alpha/2}\frac{\hat{\sigma}'_{AMR}-\hat{\sigma}_{AMR} + \hat{\sigma}_{AMR}}{\sqrt{n}} \\
    & = \underbrace{\hat{\theta}^0_{AMR} + z_{1-\alpha/2}\frac{\hat{\sigma}_{AMR}}{\sqrt{n}}}_{\text{(A$_2$)}} +  \underbrace{\hat{\theta}_{AMR} - \hat{\theta}^0_{AMR}}_{\text{(B)}} +  \underbrace{z_{1-\alpha/2}\frac{\hat{\sigma}'_{AMR}-\hat{\sigma}_{AMR}}{\sqrt{n}}}_{\text{(C)}}\;.
\end{align*}
Note (A$_1$) and (A$_2$) are the lower- and upper-bounds of (\ref{eqn:ci_oracle}). If the adjusted lower bound is no smaller than the oracle lower bound and the adjusted upper bound is no larger than the oracle upper bound, then the overall confidence interval has at least the nominal coverage. In order to ensure that the adjusted interval is (asymptotically) wider than the oracle interval, it is sufficient to have that
$\lvert$(B)$\rvert \leq \lvert$(C)$\rvert$ 
with high probability.
$$\frac{\sqrt{n}}{z_{1-\alpha/2}}\lvert(\text{C})\rvert =  \bigg\vert\sqrt{\P_n\big\{\hat{h}(A,X)\hat{Y}^* - \P_n\big[\hat{w}^*(\hat{Y}^*)\hat{Y}^*\big]\big\}^2}-\sqrt{\P_n\big\{\hat{w}^*(\hat{Y}^*)\hat{Y}^*\ - \P_n\big[\hat{w}^*(\hat{Y}^*)\hat{Y}^*\big]\big\}^2} \bigg \vert$$ 
 converges to $\Delta = \Var\{\varphi_{AIPW}(Z;\theta,\xi)\}-\Var\{\varphi_{AMR}(Z;\theta,\xi)\}$, since 
\begin{align*}
    \P_n\left\{\hat{h}(A,X)\hat{Y}^* - \P_n\left[\hat{w}^*(\hat{Y}^*)\hat{Y}^*\right]\right\}^2 &\overset{P}{\longrightarrow }\Var\{\varphi_{AIPW}(Z;\theta,\xi)\}\\
    \P_n\left\{\hat{w}^*(\hat{Y}^*)\hat{Y}^* - \P_n\left[\hat{w}^*(\hat{Y}^*)\hat{Y}^*\right]\right\}^2 &\overset{P}{\longrightarrow } \Var\{\varphi_{AMR}(Z;\theta,\xi)\}
\end{align*}
and $\Var\{\varphi_{AIPW}(Z;\theta,\xi)\} \geq \Var\{\varphi_{AMR}(Z;\theta,\xi)\}$.

On the other hand, if we have $\lvert$(B)$\rvert = \big\vert\P_n[(\hat{w}^*-w^{*0})(\hat{Y}^*)\hat{Y}^*]\big\vert = \operatorname{O}_\P(n^{-1/2})$, then under the conditions in \Cref{prop:conservative_ci} we obtain $\lvert$(B)$\rvert \leq \lvert$(C)$\rvert$. The adjusted confidence interval is thus conservative (its coverage is at least $1-\alpha$).

This condition is plausible in our setting, especially when there is severe lack of overlap (which increases $\Delta$) thereby ensuring that the additional error from replacing $\hat{\theta}^0_{AMR}$ by $\hat{\theta}_{AMR}$ is dominated by the difference in the variance estimators.
\end{proof}
\section{Weights estimation}
\label{sec:weights_estimation_models_overview}
The design of our estimator relies on the estimation of the weight function, $\hat{w}$ and $\hat{w}^{*}$. In our experiments, we tried several univariate regression methods to estimate the weights function via minimizing the $L_2$ loss defined in (\ref{eqn:loss_weights_mr}) and (\ref{eqn:loss_weights_amr}).  We  present some observations from the empirical results with different univariate regression models here:

\paragraph{Nadaraya-Watson regression \citep{nadaraya1964estimating, watson1964smooth}} Although the estimates $\hat{\theta}_{AMR}$ and $\hat{\theta}_{MR}$ generally show small bias comparable to doubly robust estimators, the variance is relatively high. This characteristic sensitivity to the choice of bandwidth regarding bias-variance trade-off, and the high variance especially near the boundary when data is sparse is a known property of the Nadaraya-Watson estimator \citep{wasserman2006all}. In the experiments we also observed that increasing the order of kernel leads to bias reduction, but variance inflation.
\paragraph{Kernel ridge regression} This method performs best in terms of MSE, providing a significant reduction compared to doubly robust estimators, particularly when the sample size is small. However, the ridge regularization term introduces bias, which can be substantial with small sample sizes. In our experiments, kernel ridge regression can achieve asymptotic unbiasedness if the regularization parameter decreases appropriately as the sample size increases. The regularization term manages the trade-off between bias and variance; as it diminishes, the estimator can converge to the true underlying function, provided it lies within the RKHS induced by the kernel.
\paragraph{Nonparametric Kernal-based Local Polynomial regression \citep{calonico2018effect}}  We use the package \texttt{nprobust} \citep{calonico2019nprobust} for implementation of this model. This method offers lower bias compared to other approaches, but it exhibits relatively higher variance, thus higher MSE of $\hat{\theta}_{AMR}$ than estimation produced by kernel ridge regression and $\hat{\theta}_{AIPW}$. The higher variance is due to fitting a separate model for each prediction point, making the method sensitive to noise and outliers. Additionally, it is computationally intensive and unstable, as it involves solving a weighted least squares problem at each prediction point.

\section{Derivation of weights in \Cref{ex:explict_weight}}
\label{sec:true_weights_derivation}
In this section we provide the derivation of weights $w(y)$, $w^0(y)$, $w^{*}(y^{*})$ and $w^{*0}(y^{*})$ in \Cref{ex:explict_weight}.  As demonstrated in \Cref{ex:explict_weight}, we assume $Y^a|X \sim \mathcal{N}(\mu^a(X),\sigma^2)$, $a\in \{0,1\}$, and $A|X \sim \textit{Bernoulli}(\pi(X))$. Further, let $\mathcal{N}(x; a, b)$ denote the density of $x$ under Gaussian distribution with mean $a$ and variance $b$. We have:
\begin{align*}
    Y& = AY^1 + (1-A)Y^0, \\
    Y^a&\mid X\sim N\left(\mu^a(X), \sigma^2\right).
\end{align*}
We also have:
\begin{align*}
    p\left(A=a,X=x\middle| Y=y\right) &= \frac{p_{AXY}(a,x,y)}{p_Y(y)}\\
    &= \frac{p_{Y\mid AX}\left(y\middle| a,x\right)\,p_{A\mid X}\left(a\middle| x\right)\,p_X(x)}{\int_\mathcal{X}\int_\mathcal{A}p_{Y\mid AX}\left(y\middle| a,x\right)\,p_{A\mid X}\left(a\middle| x\right)\,p_X(x)\,\mathrm{d}a\,\mathrm{d}x}; \\
    p\left(A=1,X=x \middle| Y=y\right) &= \pi(x)\,\mathcal{N}(y;\; \mu^1(x),\;\sigma^2)\,p_X(x)/p_Y(y)\\
    p\left(A=0,X=x \middle| Y=y\right) &= \left\{1-\pi(X)\right\}\,\mathcal{N}(y; \;\mu^0(x),\;\sigma^2)\,p_X(x)/p_Y(y)\\
    p_Y(y) &=  \int_\mathcal{X}\int_\mathcal{A}p_{Y\mid A,X}\left(y\middle| 
 a,x\right)\,p_{A\mid X}\left(a\middle| x\right)p_X(x)\,\mathrm{d}a\,\mathrm{d}x\\
    &=\int_\mathcal{X}\left[\pi(x)\mathcal{N}(y;\;\mu^1(x),\sigma^2) + \left\{1-\pi(x)\right\}\mathcal{N}(y;\;\mu^0(x),\;\sigma^2)\right]p_X(x)\,\mathrm{d}x \\
    &= \E_X\left[\pi(X)\mathcal{N}(y;\;\mu^1(X),\;\sigma^2) + \left\{1-\pi(X)\right\}\mathcal{N}(y;\;\mu^0(X),\;\sigma^2)\right].
\end{align*}
%
%
We then find:
\begin{align*}
    w^0(y)&=\E\left[\frac{A-\hat{\pi}(X)}{\hat{\pi}(X)\left\{1-\hat{\pi}(X)\right\}}\middle| Y=y\right]\\
    &= \int_\mathcal{X}\int_\mathcal{A} \frac{a-\hat{\pi}(x)}{\hat{\pi}(x)\left\{1-\hat{\pi}(X)\right\}}p\left(a,x\middle| Y=y\right)\mathrm{d}x\mathrm{d}a\\
    &= \frac{1}{p_Y(y)}\int_\mathcal{X} \left\{\frac{1}{\hat{\pi}(x)}p\left(A=1,X=x\middle| Y=y\right) -\frac{1}{1-\hat{\pi}(x)} p\left(A=0,X=x\middle| Y=y\right) \right\} \mathrm{d}x\\
    &= \frac{1}{p_Y(y)}\int_\mathcal{X}\left\{\frac{\pi(x)}{\hat{\pi}(x)}\mathcal{N}\big(y;\; \mu^1(x),\;\sigma^2\big)-\frac{1-\pi(x)}{1-\hat{\pi}(x)}\mathcal{N}\big(y;\; \mu^0(x),\;\sigma^2\big)\right\}p_X(x) \, \mathrm{d}x\\
    &= \frac{\E_X\left[\frac{\pi(x)}{\hat{\pi}(x)}\mathcal{N}\big(y;\; \mu^1(X),\;\sigma^2\big) - \frac{1-\pi(x)}{1-\hat{\pi}(x)}\mathcal{N}\big(y;\; \mu^0(X),\;\sigma^2\big)\right]}{ \E_X\left[\pi(X)\mathcal{N}\big(y;\;\mu^1(X),\;\sigma^2\big) + \left\{1-\pi(X)\right\}\mathcal{N}\big(y;\;\mu^0(X),\;\sigma^2\big)\right]}.
\end{align*}
In the case where we have an oracle estimator $\hat\pi = \pi$, we obtain
\begin{align*}
    w(y)
    &= \frac{\E_X\left[\mathcal{N}\big(y;\; \hat\mu^1(X),\;\sigma^2\big) - \mathcal{N}\big(y;\; \hat\mu^0(X),\;\sigma^2\big)\right]}{ \E_X\left[\pi(X)\mathcal{N}\big(y;\;\hat\mu^1(X),\;\sigma^2\big) + \left\{1-\pi(X)\right\}\mathcal{N}\big(y;\;\hat\mu^0(X),\;\sigma^2\big)\right]},
\end{align*}
which is clearly different from $w^0(y)$.

Similarly, we provide the derivation of AMR weights. With $\hat{Y}^{*} = Y-\hat{\mu}^{*}$, we have:
\begin{align*}
    \hat{Y}^{*} &= AY^1 + (1-A)Y^0 -\left[\hat{\pi}(X)\hat{\mu}^0(X) + \left\{1-\hat{\pi}(X)\right\}\hat{\mu}^1(X)\right]\\
    \hat{Y}^{*}&\mid A=1, X \sim \mathcal{N}\big(\mu^1(X) -\hat{\mu}^{*}, \sigma^2\big)\\
    \hat{Y}^{*}&\mid A=0, X \sim \mathcal{N}\big(\mu^0(X) -\hat{\mu}^{*}, \sigma^2\big).
\end{align*}
We also have:
\begin{align*}
    p(A=a,X=x|\hat{Y}^{*}=y^*) &= \frac{p_{AX\hat{Y}^*}(a,x,y^{*})}{p_{\hat{Y}^{*}}(y^{*})}\\
    &= \frac{p_{\hat{Y}^{*} \mid AX}(y^* \mid a,x)\,p_{A \mid X}(a\mid x)\,p_X(x)}{\int_\mathcal{X}\int_\mathcal{A}p_{\hat{Y}^{*} \mid AX}(y^*\mid a,x)\,p_{A\mid X}(a \mid x)\,p_X(x)\mathrm{d}a\,\mathrm{d}x} \\
    p(A=1,X=x \mid \hat{Y}^{*}=y^{*}) &= \pi(x)\mathcal{N}\big(y^{*}; \mu^1(x) - \hat{\mu}^{*},\sigma^2\big)\,p_X(x)/p_{\hat{Y}^{*}}(y^{*})\\
    p(A=0,X=x \mid \hat{Y}^{*}=y^{*}) &= \left\{1-\pi(x)\right\}\mathcal{N}\big(y^{*}; \mu^0(x) -\hat{\mu}^{*},\sigma^2\big)\,p_X(x)/p_{\hat{Y}^{*}}(y^{*}),
\end{align*}
where 
\begin{align*}
p_{\hat{Y}^{*}}(y^{*}) &=  \int_\mathcal{X}\int_\mathcal{A}p_{\hat{Y}^{*} \mid AX}(y^*\mid a,x)\,p_{A\mid X}(a \mid x)\,p_X(x)\mathrm{d}a\,\mathrm{d}x\\
    &=\int_\mathcal{X}\left[\pi(x)\mathcal{N}\big(y^{*};\mu^1(x) -\hat{\mu}^{*},\sigma^2\big) + \left\{1-\pi(x)\right\}\mathcal{N}\big(y^{*};\mu^0(x) -\hat{\mu}^{*},\sigma^2\big)\right]\,p_X(x)\mathrm{d}x\\
    &= \E_X\Bigl[\pi(X)\mathcal{N}\big(y^{*};\mu^1(X) -\hat{\mu}^{*},\sigma^2\big) + \left\{1-\pi(X)\right\}\mathcal{N}\big(y^{*};\mu^0(X) -\hat{\mu}^{*},\sigma^2\big)
    \Bigr].
\end{align*}
We thus obtain:
\small
\begin{align*}
    w^{*0}(y^{*}) &= \E\left[\frac{A-\hat{\pi}(X)}{\hat{\pi}(X)\left\{1-\hat{\pi}(X)\right\}}\middle|\hat{Y}^{*}=y^{*}\right] \\
    &= \int_\mathcal{X}\int_\mathcal{A} \frac{a-\hat{\pi}(x)}{\hat{\pi}(x)\left\{1-\hat{\pi}(X)\right\}}p_{AX\mid \hat{Y}^*}\left(a,x \middle|y^{*}\right)\mathrm{d}x\mathrm{d}a\\
    &= \frac{1}{p_{\hat{Y}^{*}}(y^{*})}\int_\mathcal{X}\left[\frac{1}{\hat{\pi}(x)}p\left(A=1,X=x\middle|\hat{Y}^{*}=y^{*}\right) - \frac{1}{1-\hat{\pi}(x)}p\left(A=0,X=x\middle|\hat{Y}^{*}=y^{*}\right)\right]\mathrm{d}x\\
    &=  \frac{\E_X\left[\frac{\pi(X)}{\hat{\pi}(X)}\mathcal{N}\big(y^{*};\mu^1(X) - \hat{\mu}^{*},\sigma^2\big) - \frac{1-\pi(X)}{1-\hat{\pi}(X)}\mathcal{N}\big(y^{*};\mu^0(X) - \hat{\mu}^{*},\sigma^2\big)\right]}{ \E_X\Bigl[\pi(X)\mathcal{N}\big(y^{*};\mu^1(X) - \hat{\mu}^{*},\sigma^2\big)+ \left\{1-\pi(X)\right\}\mathcal{N}\big(y^{*};\mu^0(X) - \hat{\mu}^{*},\sigma^2\big)\Bigr]}.
\end{align*}
\normalsize

When we have oracle estimators $\hat{\pi}=\pi, \hat{\mu}^a=\mu^a$ for $a\in\{0,1\}$, we find
\begin{align*}
    w^{*}(y^{*})  =
    \frac{\E_X\Bigl[\mathcal{N}\big(y^{*};\pi(X)\tau(X),\sigma^2\big) - \mathcal{N}\big(y^{*}; -\left\{1-\pi(X)\right\}\tau(X),\sigma^2\big)\Bigr]}{ \E_X\Bigl[\pi(X)\mathcal{N}\big(y^{*};\pi(X)\tau(X),\sigma^2\big)+ \left\{1-\pi(X)\right\}\mathcal{N}\big(y^{*};-\left\{1-\pi(X)\right\}\tau(X),\sigma^2\big)\Bigr]},
\end{align*} where $\tau(x) = \mu^1(x)-\mu^0(x)$. Clearly, $w^*$ and $w^{*}$ are different, since they take the conditional expectation on different variable spaces---$w^*$ is on $\sigma(Y^*)$, while $w^{*0}$ is on $\sigma(\hat{Y}^{*})$. 

\section{Additional experiment: Amazon review}\label{sec:amazon}
In this section, we compare the performance of our method with the Treatment Ignorant (TI) estimator proposed by \citet{gui2022causal} on the publicly available Amazon review dataset. We mainly use reviews for mp3, CD and vinyl, setting the treatment $A$ as whether the review is five stars ($A=1$) or one/two stars ($A=0$). We use the same semi-synthetic data generation setting as created in \citet{gui2022causal} to simulate outcome $Y$, that is, we additionally set $C=1$ if the product is a CD and $C=0$ otherwise, $Y=\beta_a A+\beta_c(\pi(C) - \beta_o)+\gamma \epsilon$, $\epsilon \sim \mathcal{N}(0,1)$. Here $\beta_a, \beta_c, \beta_o$ stand for the average treatment effect, the strength of confounding and the offset term for $\E\pi(C)$, where $\pi(C) = P(A=1|C)$ is the propensity score. The dataset has 10,658 entries in total. 

The TI estimator addresses the challenges of insufficient overlap in causal effect estimation on text data, by estimating propensity scores conditional on the potential outcomes, i.e.~$\hat{\pi}=\E\left[ A \middle| \hat{\mu}^0,\hat{\mu}^1\right]$. Specifically, the authors propose a transformer-based neural network architecture similar to the DragonNet, that first learns the potential outcomes leveraging the propensity score estimation information by loss function design, and then finalize the propensity scores estimation using the estimated potential outcomes as features. This approach closely resembles \citet{benkeser2020nonparametric} but is tailored for text-based applications in that TI estimator is a transformer, attention-based architecture.

To ensure a fair comparison with the TI estimator, we used BERT to produce embeddings as covariates for each Amazon review. Specifically, we pooled token embeddings for each review by using the \texttt{[CLS]} embedding. These embeddings were then fed into a 3-layer feedforward neural network model to estimate $\hat{\pi}$. For $\hat{\mu}^a, a\in \{0,1\}$, we keep them the same as those used in the TI estimator.

We first provide results in comparison with TI estimator under its recommended configurations in model building and data simulation ($\beta_a=2,\;\beta_c=50,\;\pi(0)=0.8,\;\pi(1)=0.6$) in \Cref{fig:default_music} and \Cref{tab:performance_default_music}. We randomly sampled 1000 data points from the dataset. Although we did not perform any hyperparameter tuning for our nuisance parameter models, the performance of AMR and MR remains comparable.

\begin{figure}
    \centering
    \includegraphics[width=0.8\linewidth]{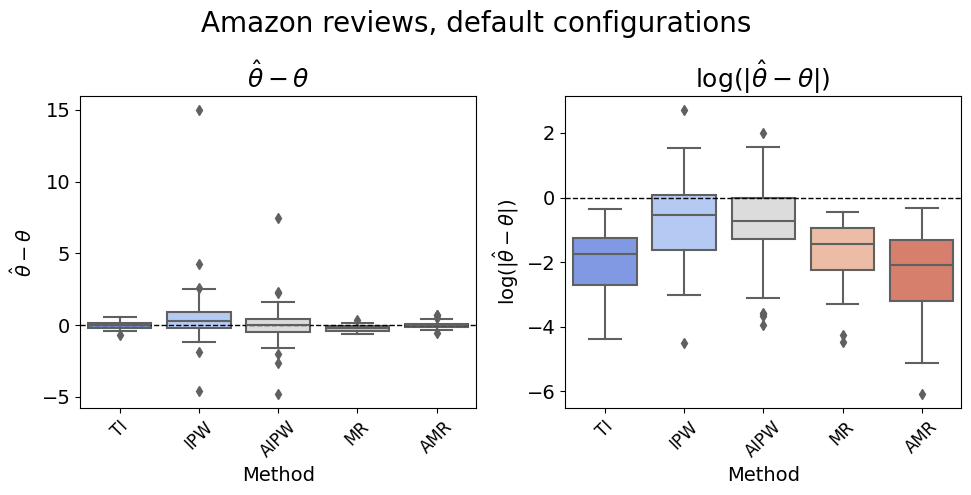}
    \caption{Estimator performances on $K=50$ draws of  \textit{Amazon Reviews} dataset, $n=1000$, default configuration.}
    \label{fig:default_music}
\end{figure}

\begin{table}
\begin{center}
\begin{tabular}{rrrrrr}
\toprule
     & \textbf{TI}& \textbf{IPW} & \textbf{AIPW} & \textbf{MR} & \textbf{AMR}\\
\midrule
     Bias & $-0.0300$ & $0.622$ & $-0.0634$ & $-0.231$ &$\textbf{0.0201}$ \\

    MAE & $0.192$ & $1.12$ & $0.901$ & $0.252$ & $\textbf{0.182}$ \\

   RMSE & $0.252$ & $2.45$ & $1.55$ & $0.312$ & $\textbf{0.251}$\\
\bottomrule
\end{tabular}
\end{center}
\caption{Estimator performances on $K=50$ draws of  \textit{Amazon Reviews} dataset, $n=1000$, default configuration, $\theta=2$.}
\label{tab:performance_default_music}
\end{table} 

However, the main limitation of the TI estimator is apparent---it hinges on accurately estimating potential outcomes. Under the default configuration recommended by the authors, these estimates are very close to the true underlying values, explaining the strong performance of the TI. We present results with the data simulation changed by introducing extra interaction term in the true potential outcome (\Cref{tab:performance_new_df_music}) and replacing the transformer-based architecture in TI with linear regression for potential outcome estimation (\Cref{tab:performance_replaced_music}). TI's performance degrades immediately in both scenarios, aligning with our conjecture about its strong dependence on accurate outcome estimation. On the contrary, the performances of MR and AMR remain relatively stable. This result further underscores the double robustness property of AMR. Even under model misspecification of $\hat{\mu}^a$, AMR produces relatively accurate estimations, saving a lot of computation time each run compared to TI without sacrificing the performance.

\begin{table}
\begin{center}
\begin{tabular}{rrrrrr}
\toprule
     & \textbf{TI}& \textbf{IPW} & \textbf{AIPW} & \textbf{MR} & \textbf{AMR}\\
\midrule
    Bias & $0.360$ & $-13.5$ & $-0.150$ & $1.82$ & $\textbf{0.000112}$\\

    MAE & $1.05$ & $16.7$ & $16.3$ & $1.87$ & $\textbf{0.723}$ \\

    RMSE & $1.92$  & $30.5$ & $27.3$ & $1.92$ & $\textbf{1.06}$\\
\bottomrule
\end{tabular}
\end{center}
\caption{Estimator performances on $K=50$ draws of  \textit{Amazon Reviews} dataset, $n=1000$, data simulation changed, $\theta=0.2$.}
\label{tab:performance_new_df_music}
\end{table}

\begin{table}
\begin{center}
\begin{tabular}{rrrrrr}
\toprule
     & \textbf{TI}& \textbf{IPW} & \textbf{AIPW} & \textbf{MR} & \textbf{AMR}\\
\midrule
    Bias & $-2.58\times 10^6$ & $0.372$ & $5.23\times 10^7$ &$\textbf{-0.221}$ &$0.803$\\

    MAE & $~2.58\times 10^6$ & $-0.791$ & $6.31 \times 10^7$ & $\textbf{~0.242}$ & $2.53$ \\

    RMSE &  $~1.13\times 10^7$  & $1.37$ & $3.75\times 10^7$ & $\textbf{~0.283}$ & $9.43$\\
\bottomrule
\end{tabular}
\end{center}
\caption{Estimator performances on $K=50$ draws of  \textit{Amazon Reviews} dataset, $n=1000$, linear model for potential outcome, $\theta=3$.}
\label{tab:performance_replaced_music}
\end{table} 

\section{Analogue to e-score}
\label{sec:e-score}
We propose two estimators following similar design idea of AMR and e-score, which solely considers re-weighting the bias $Y-\mu^1$. Same as in \citet{diaz2018doubly}, we derive this population estimator for ATT.
\begin{align*}
    \theta^{AMR-e'}_{ATT} &= \E\left[\frac{A}{\pi(X)}(Y-\mu^1(X)) + \mu^1(X)\right]\\
    &=  \E\left[\E\left[\frac{A}{\pi(X)}\middle | (Y-\mu^1(X))\right] (Y-\mu^1(X)) + \mu^1(X)\right]\;\\
    &\coloneqq \E\left[\frac{1}{e'}(Y-\mu^1(X)) + \mu^1(X)\right],
\end{align*}
where $e'=1/\E\left[\frac{A}{\pi(X)}\middle| (Y-\mu^1(X))\right] $. 

Or, we can leverages e-score's first stage regression, $r=\E\left[Y-\mu^1(X) \middle| A=1,\pi(X) \right]$:
\begin{align*}
    \theta^{AMR-e}_{ATT} &= \E\left[\frac{A}{\pi(X)}(Y-\mu^1(X)) + \mu^1(X)\right]\\
     &= \E\left[\E\left[\frac{A}{\pi(X)}(Y-\mu^1(X)) + \mu^1(X) \middle| A=1, \pi(X)\right]\right]\\
     & = \E\left[\frac{A}{\pi(X)}\E\left[(Y-\mu^1(X))\middle| A=1, \pi(X)\right]+ \E\left[\mu^1 \middle| A=1, \pi(X)\right]\right]\\
     & =  \E\left[\frac{A}{\pi(X)}r  \right] + \mu^1(X)\\
     & =  \E\left[\E\left[\frac{A}{\pi(X)}\middle| r\right]r \right] + \mu^1(X)\\
     & \coloneqq \E\left[\frac{1}{e}r  + \mu^1(X)\right].
\end{align*}
where $e = 1/\E\left[\frac{A}{\pi(X)}\middle|r\right]$. 
Apparently, the second version would have the same problem as the original e-score where there are more regression steps (and decision points), which might result in higher risk of propagating errors.


\end{document}